\documentclass[a4paper,10pt]{article}
\usepackage{fullpage}
\usepackage{graphicx}
\usepackage[T1]{fontenc}
\usepackage{epsfig}
\usepackage[noadjust]{cite}
\usepackage{color}
\usepackage{enumerate}
\usepackage{url}
\usepackage{comment}
\usepackage{hyperref}
\usepackage{mathtools}
\usepackage{breqn}
\usepackage{enumitem}
\usepackage{amsmath,amssymb,amstext, amsthm}
\usepackage{thmtools}
\usepackage{amsbsy,bm}
\usepackage{amsfonts}
\usepackage{algorithm2e}
\usepackage{calc}
\usepackage{lineno}
\newcommand{\customdash}[1]{\rule[0.5ex]{#1}{0.4pt}}
\newcommand{\rank}{\textsf{rank}}
\newcommand{\select}{\textsf{select}}
\newcommand{\findopen}{\textsf{findopen}}
\newcommand{\findclose}{\textsf{findclose}}

\newcommand{\MinC}{\textsf{MinC}}
\newcommand{\MaxC}{\textsf{MaxC}}
\newcommand{\lr}{\textsf{lr}}
\newcommand{\lp}{\textsf{lp}}
\newcommand{\rp}{\textsf{rp}}
\newcommand{\lrp}{\textsf{lrp}}
\newcommand{\dabove}{\textsf{above}}
\newcommand{\dbelow}{\textsf{below}}
\newcommand{\dleft}{\textsf{left}}
\newcommand{\dright}{\textsf{right}}
\newcommand{\daboveset}{\textsf{above\_set}}
\newcommand{\dbelowset}{\textsf{below\_set}}
\newcommand{\dleftset}{\textsf{left\_set}}
\newcommand{\drightset}{\textsf{right\_set}}
\newcommand{\psv}{\textsf{PSV}}
\newcommand{\plv}{\textsf{PLV}}
\newcommand{\nsv}{\textsf{NSV}}
\newcommand{\nlv}{\textsf{NLV}}
\newcommand{\rminq}{\textsf{RMin}}
\newcommand{\rmaxq}{\textsf{RMax}}

\newcommand{\pl}{\textsf{parent\_label}}
\newcommand{\rcl}{\textsf{right\_child\_label}}
\newcommand{\lcl}{\textsf{left\_child\_label}}
\newcommand{\nxt}{\textsf{next}}
\newcommand{\RNum}[1]{\uppercase\expandafter{\romannumeral #1\relax}}

\newtheorem{theorem}{Theorem}
\newtheorem{corollary}{Corollary}
\newtheorem{lemma}{Lemma}
\newtheorem{proposition}{Proposition}
\theoremstyle{definition}
\newtheorem{example}{Example}

\title{Succinct Data Structures for Baxter Permutation and Related Families}
\author{
    Sankardeep Chakraborty \\
    The University of Tokyo, Japan \\
    sankardeep.chakraborty@gmail.com
    \and
    Seungbum Jo \\
    Chungnam National University, South Korea \\
    sbjo@cnu.ac.kr
    \and
    Geunho Kim \\
    Pohang University of Science and Technology \\
    gnhokim@postech.ac.kr
    \and
    Kunihiko Sadakane \\
    The University of Tokyo, Japan \\
    sada@mist.i.u-tokyo.ac.jp
}
\date{}
\begin{document}
\maketitle

\begin{abstract}
A permutation $\pi: [n] \rightarrow [n]$ is a Baxter permutation if and only if it does not contain either of the patterns $2\customdash{0.4em}41\customdash{0.4em}3$ and $3\customdash{0.4em}14\customdash{0.4em}2$. Baxter permutations are one of the most widely studied subclasses of general permutation due to their connections with various combinatorial objects such as plane bipolar orientations and mosaic floorplans, etc. In this paper, we introduce a novel succinct representation (i.e., using $o(n)$ additional bits from their information-theoretical lower bounds) for Baxter permutations of size $n$ that supports $\pi(i)$ and $\pi^{-1}(j)$ queries for any $i \in [n]$ in $O(f_1(n))$ and $O(f_2(n))$ time, respectively. Here, $f_1(n)$ and $f_2(n)$ are arbitrary increasing functions that satisfy the conditions $\omega(\log n)$ and $\omega(\log^2 n)$, respectively. This stands out as the first succinct representation with sub-linear worst-case query times for Baxter permutations. 
The main idea is to traverse the Cartesian tree on the permutation using a simple yet elegant \textit{two-stack algorithm} which traverses the nodes in ascending order of their corresponding labels and stores the necessary information throughout the algorithm. 

Additionally, we consider a subclass of Baxter permutations called \textit{separable permutations}, which do not contain either of the patterns $2\customdash{0.4em}4\customdash{0.4em}1\customdash{0.4em}3$ and $3\customdash{0.4em}1\customdash{0.4em}4\customdash{0.4em}2$. 
In this paper, we provide the first succinct representation of the separable permutation $\rho: [n] \rightarrow [n]$ of size $n$ that supports both $\rho(i)$ and $\rho^{-1}(j)$ queries in $O(1)$ time. In particular, this result circumvents Golynski's [SODA 2009] lower bound result for trade-offs between redundancy and $\rho(i)$ and $\rho^{-1}(j)$ queries. 
 
Moreover, as applications of these permutations with the queries, we also introduce the first succinct representations for mosaic/slicing floorplans, and plane bipolar orientations, which can further support specific navigational queries on them efficiently.
\end{abstract}

\section{Introduction}
A permutation $\pi: [n] \rightarrow [n]$  is a Baxter permutation if and only if there are no three indices $i<j<k$ that satisfy $\pi(j+1) < \pi(i) < \pi(k) < \pi(j)$ or $\pi(j) < \pi(k) < \pi(i) < \pi(j+1)$ (that is, $\pi$ does not have pattern $2\customdash{0.4em}41\customdash{0.4em}3$ or $3\customdash{0.4em}14\customdash{0.4em}2$)~\cite{Baxter:original}. For example, $3~5~2~1~4$ is not a Baxter permutation because the pattern $2\customdash{0.4em}41\customdash{0.4em}3$ appears ($\pi(2+1)=2 < \pi(1)=3 < \pi(5)=4 < \pi(2)=5$ holds). A Baxter permutation $\pi$ is \textit{alternating} if the elements in $\pi$ rise and descend alternately. One can also consider \textit{separable permutations}, which are defined as the permutations without two patterns $2\customdash{0.4em}4\customdash{0.4em}1\customdash{0.4em}3$ and $3\customdash{0.4em}1\customdash{0.4em}4\customdash{0.4em}2$~\cite{DBLP:journals/ipl/BoseBL98}.
From the definitions, any separable permutation
is also a Baxter permutation, but the converse does not hold. For example, $2~5~6~3~1~4~8~7$ is a Baxter permutation but not a separable permutation because of the appearance of the pattern $2\customdash{0.4em}4\customdash{0.4em}1\customdash{0.4em}3$ ($2~5~1~4$). 

In this paper, we focus on the design of a succinct data structure for a Baxter permutation $\pi$ of size $n$, 
i.e., the data structure that uses up to $o(n)$ extra bits in addition to the information-theoretical lower bound along with supporting relevant queries efficiently. Mainly, we consider the following two fundamental queries on $\pi$: (1) $\pi(i)$ returns the $i$-th value of $\pi$, and (2) $\pi^{-1}(j)$ returns the index $i$ of $\pi(i) = j$. We also consider the design of a succinct data structure for a separable permutation $\rho$ of size $n$ that supports $\rho(i)$ and $\rho^{-1}(j)$ queries. 
In the rest of this paper, $\log$ denotes the logarithm to the base $2$, and we assume a word-RAM model with  $\Theta(\log n)$-bit word size, where $n$ is the size of the input. Also, we ignore all ceiling and floor operations that do not impact the final results.

\subsection{Previous Results}
For general permutations, there exist upper and lower bound results for succinct data structures supporting both $\pi(i)$ and $\pi^{-1}(j)$ queries in sub-linear time~\cite{munro2012succinct, golynski2009cell}. However, to the best of our knowledge, there does not exist any data structures for efficiently supporting these queries on any subclass of general permutations. One can consider suffix arrays~\cite{DBLP:journals/siamcomp/GrossiV05} as a subclass of general permutations, but their space consumption majorly depends on the entropy of input strings. This implies that for certain input strings, $\Omega(n \log n)$ bits (asymptotically the same space needed for storing general permutations) are necessary for storing the suffix arrays on them.

Baxter permutation is one of the most widely studied classes of permutations~\cite{bona} because diverse combinatorial objects, for example, plane bipolar orientations, mosaic floorplans, twin pairs of binary trees, etc. have a bijection with Baxter permutations~\cite{DBLP:journals/dam/AckermanBP06, DBLP:journals/jct/FelsnerFNO11}. Note that some of these objects are used in many applied areas. For example, mosaic floorplans are used in large-scale chip design~\cite{lengauer2012combinatorial}, plane bipolar orientations are used to draw graphs in various flavors (visibility~\cite{TamassiaT86}, straight-line drawing~\cite{Fusy06}), and floorplan partitioning is used to design a model for stochastic processes~\cite{nakano2020baxter}.
The number of distinct Baxter permutations of size $n$ is $\Theta(8^n/n^4)$~\cite{shen2003bounds}, which implies that at least $3n-o(n)$ bits are necessary to store a Baxter permutation of size $n$. Furthermore, the number of distinct alternating Baxter permutations of size $2n$ (resp. $2n+1$) is $(c_n)^2$ (resp. $c_n c_{n+1}$) where $c_n = \frac{(2n)!}{(n+1)!n!}$ is the $n$-th Catalan number~\cite{cori1986shuffle}. Therefore, at least $2n-o(n)$ bits are necessary to store an alternating Baxter permutation of size $n$. 
Dulucq and Guibert~\cite{dulucq1996stack} established a bijection between Baxter permutations $\pi$ of size $n$ and a pair of unlabeled binary trees, called \textit{twin binary trees}, which are essentially equivalent to the pair of unlabeled minimum and maximum Cartesian trees~\cite{vuillemin1980unifying} for $\pi$. They provided methods for constructing $\pi$ from the structure of twin binary trees and vice versa, both of which require $O(n)$ time. Furthermore, they presented a representation scheme that requires at most $8n$ bits for Baxter permutations of size $n$ and $4n$ bits for alternating Baxter permutations of size $n$. Gawrychowski and Nicholson proposed a $3n$-bit representation that stores the tree structures of alternating representations of both minimum and maximum Cartesian trees~\cite{gawrychowski2015optimal}.
Based on the bijection established in~\cite{dulucq1996stack}, the representation in~\cite{gawrychowski2015optimal} gives a succinct representation of a Baxter permutation of size $n$. Moreover, this representation can efficiently support a wide range of tree navigational queries on these trees in $O(1)$ time using only $o(n)$ additional bits. 
However, surprisingly, all of these previous representations of $\pi$ crucially fail to address both, perhaps the most natural, $\pi(i)$ and $\pi^{-1}(j)$ queries efficiently as these queries have a worst-case time complexity of $\Theta(n)$.

Separable permutation was introduced by Bose et al.\cite{DBLP:journals/ipl/BoseBL98} as a specific case of patterns for the permutation matching problem. It is known that the number of separable permutations of size $n$ equals the \textit{large Schr{\"o}der number} $A_n$, which is $\Theta\left(\frac{(3+2\sqrt{2})^n}{n^{1.5}}\right)$\cite{YaoCCG03}. Consequently, to store a separable permutation $\rho$ of size $n$, at least $n \log (3+2\sqrt{2}) - O(\log n) \simeq 2.54n - O(\log n)$ bits are necessary. Bose et al.~\cite{DBLP:journals/ipl/BoseBL98} also showed that $\rho$ can be encoded as a \textit{separable tree}, which is a labeled tree with at most $2n-1$ nodes. Thus, by storing the separable tree using $O(n \log n)$ bits, one can support both $\rho(i)$ and $\rho^{-1}(j)$ queries in $O(1)$ time using standard tree navigation queries. Yao et al.~\cite{YaoCCG03} showed a bijection between all canonical forms of separable trees with $n$ leaves and the separable permutations of size $n$. To the best of our knowledge, there exists no $o(n \log n)$-bit representation for storing either separable permutations or their corresponding separable trees that can be constructed in polynomial time while supporting $\rho$ queries in sub-linear time.

A mosaic floorplan is a collection of rectangular objects that partition a single rectangular region. Due to its broad range of applications, there is a long history of results (see~\cite{he2014simple,YaoCCG03} and the references therein) concerning the representation of mosaic floorplans of size $n$ in small space~\cite{hong2000corner, DBLP:journals/dam/AckermanBP06, he2014simple}. 
Ackerman et al.~\cite{DBLP:journals/dam/AckermanBP06} presented a linear-time algorithm to construct a mosaic floorplan of size $n$ from its corresponding Baxter permutation of size $n$ and vice versa. Building on this construction algorithm, He~\cite{he2014simple} proposed the current state-of-the-art, a succinct representation of a mosaic floorplan of size $n$ using $3n-3$ bits. Again, all of these previous representations primarily focus on constructing a complete mosaic floorplan structure and do not consider supporting navigational queries, e.g., return a rectangular object immediately adjacent to the query object in terms of being left, right, above, or below it, without constructing it completely. Note that these queries have strong applications like the placement of blocks on the chip~\cite{DBLP:journals/dam/AckermanBP06, floorplanchip}. 
There also exists a subclass of mosaic floorplans known as \textit{slicing floorplans}, which are mosaic floorplans whose rectangular objects are generated by recursively dividing a single rectangle region either horizontally or vertically. The simplicity of a slicing floorplan makes it an efficient solution for optimization problems, as stated in~\cite{1346395}.
Yao et al.~\cite{YaoCCG03} showed there exists a bijection between separable permutations of size $n$ and slicing floorplans with $n$ rectangular objects. They also showed that separable trees can be used to represent the positions of rectangular objects in the corresponding slicing floorplans. However, to the best of our knowledge, there exists no representation of a slicing floorplan using $o(n \log n)$ bits that supports the above queries without reconstructing it.

\subsection{Our Results and Main Idea}
In this paper, we first introduce a $(3n+o(n))$-bit representation of a Baxter permutation $\pi$ of size $n$ that can support $\pi(i)$ and $\pi^{-1}(j)$ queries in $O(f_1(n))$ and $O(f_2(n))$ time respectively. Here, $f_1(n)$ and $f_2(n)$ are any increasing functions that satisfy $\omega(\log n)$ and $\omega(\log^2 n)$, respectively. We also show that the same representation provides a $(2n+o(n))$-bit representation of an alternating Baxter permutation of size $n$ with the same query times. These are the first succinct representations of Baxter and alternating Baxter permutations that can support the queries in sub-linear time in the worst case.

Our main idea of the representation is as follows. To represent $\pi$, it suffices to store the minimum or maximum Cartesian tree defined on $\pi$ along with their labels. Here the main challenging part is to decode the label of any node in either of the trees in sub-linear time, using $o(n)$-bit auxiliary structures. Note that all the previous representations either require linear time for the decoding or explicitly store the labels using $O(n \log n)$ bits.
To address this issue, we first introduce an algorithm that labels the nodes in the minimum Cartesian tree 
in ascending order of their labels. This algorithm employs two stacks and only requires information on whether each node with label $i$ is a left or right child of its parent, as well as whether it has left and/or right children. Note that unlike the algorithm of~\cite{dulucq1996stack}, our algorithm does not use the structure of the maximum Cartesian tree. We then proceed to construct a representation using at most $3n+o(n)$ bits, which stores the information used throughout our labeling algorithm. We show that this representation can decode the minimum Cartesian tree, including the labels on its nodes. This approach was not considered in previous succinct representations that focused on storing the tree structures of both minimum and maximum Cartesian trees, or their variants.
To support the queries efficiently, we show that given any label of a node in the minimum or maximum Cartesian tree, our representation can decode the labels of its parent, left child, and right child in $O(1)$ time with $o(n)$-bit auxiliary structures. Consequently, we can decode any $O(\log n)$-size substring of the balanced parentheses of both  minimum and maximum Cartesian trees with dummy nodes to locate nodes according to their inorder traversal (see Section~\ref{sec:inorder} for a detailed definition of the inorder traversal) on $\pi$ in $O(f_1(n))$ time. 
This decoding step plays a key role in our query algorithms, which can be achieved from non-trivial properties of our representation, and minimum and maximum Cartesian trees on Baxter permutations.
As a result, our representation not only supports $\pi(i)$ and $\pi^{-1}(j)$ queries, but also supports range minimum/maximum and previous/next larger/smaller value queries efficiently.

Next, we give a succinct representation of separable permutation $\rho$ of size $n$, which supports all the operations above in $O(1)$ time. 
Our result implies the Golynski's lower bound result~\cite{golynski2009cell} for trade-offs between redundancy and $\rho(i)$ and $\rho^{-1}(j)$ queries does not hold in separable permutations.
The main idea of the representation is to store the separable tree of $\rho$ using the \textit{tree covering algorithm}~\cite{Farzan2014}, where each micro-tree is stored as its corresponding separable permutation to achieve succinct space. Note that a similar approach has been employed for succinct representations on some graph classes~\cite{DBLP:conf/cpm/BlellochF10, DBLP:conf/dcc/ChakrabortyJSS21}. However, due to the different structure of the separable tree compared to the Cartesian tree, the utilization of non-trivial auxiliary structures is crucial for achieving $O(1)$ query time on the representation.

Finally, as applications of our succinct representations of Baxter and separable permutations, 
we present succinct data structures of mosaic and slicing floorplans and plane bipolar orientations that support various navigational queries on them efficiently. 
While construction algorithms for these structures already exist from their corresponding Baxter or separable permutations~\cite{DBLP:journals/dam/AckermanBP06, BonichonBF08},
we show that the navigational queries can be answered using a constant number of $\pi(i)$ (or $\rho(i)$), range minimum/maximum, and previous/next smaller/larger value queries on their respective permutations, which also require some nontrivial observations from the construction algorithms.
This implies that our succinct representations allow for the first time succinct representations of these structures that support various navigation queries on them in sub-linear time. 
For example, we consider two queries on mosaic and slicing floorplans as (1) checking whether two rectangular objects are adjacent, and (2) reporting all rectangular objects adjacent to the given rectangular object. 
Note that the query of (2) was previously addressed in~\cite{DBLP:journals/dam/AckermanBP06}, as the \textit{direct relation set} (DRS) query, which was computed in $O(n)$ time, and important for the actual placement of the blocks on the chip. 

The paper is organized as follows. After introducing some preliminaries in Section~\ref{pre:BP}, we introduce the representation of a Baxter permutation $\pi$ of size $n$ in Section~\ref{sec:rep}. 
In Section~\ref{sec:query}, we explain how to support $\pi(i)$ and $\pi^{-1}(i)$ queries on $\pi$, in addition to tree navigational queries on both the minimum and maximum Cartesian trees. In Section~\ref{sec:separable_represent}, we present a succinct representation of separable permutation $\rho$ that can support $\rho(i)$ and $\rho^{-1}(j)$ in $O(1)$ time.
Finally, in Section~\ref{sec:application}, we show how our representations of $\pi$ and $\rho$ can be applied to construct succinct representations for mosaic/slicing floorplans and plane bipolar orientations, while efficiently supporting certain navigational queries.

\section{Preliminaries}\label{pre:BP}
In this section, we introduce some preliminaries that will be used in the rest of the paper.
\\\\
\noindent\textbf{Cartesian trees.} 
Given a sequence $S = (s_1, s_2, \dots, s_n)$ of size $n$ from a total order, a \textit{minimum Cartesian tree} of $S$, denoted as $\MinC{}(S)$ is 
a binary tree constructed as follows~\cite{vuillemin1980unifying}:
(a) the root of the $\MinC{}(S)$ is labeled as the minimum element in $S$ (b) if the label of the root is $s_i$, the left and right subtree of $S$ are $\MinC{}(S_1)$ and $\MinC{}(S_2)$, respectively where $S_1 = (s_1, s_2, \dots, s_{i-1})$ and $S_2 = (s_{i+1}, s_{i+2}, \dots, s_n)$. 
One can also define a \text{maximum Cartesian tree} of $S$ (denoted as $\MaxC{}(S)$) analogously.
From the definition, in both $\MinC{}(S)$ and $\MaxC{}(S)$, any node with inorder $i$ is labeled with $s_i$.
\\\\
\noindent\textbf{Balanced parentheses.} \label{pre:BP2}
Given an ordered tree $T$ of $n$ nodes, the BP of $T$ (denoted as $BP(T)$) is defined as a sequence of open and closed parentheses constructed as follows~\cite{munro2001succinct}. 
One traverses $T$ from the root node in depth-first search (DFS) order. During the traversal, for each node $p \in T$, we append `(' when we visit the node $p$ for the first time, and append `)' when all the nodes on the subtree rooted at $p$ are visited, and we leave the node $p$. From the construction, it is clear that the size of $BP(T)$ is $2n$ bits, and always balanced. 
Munro and Raman~\cite{munro2001succinct} showed that both (a) $\findopen{}(i)$: returns the position of matching open parenthesis of the close parenthesis at $i$, and (b) $\findclose{}(i)$: returns the position of matching close parenthesis of the open parenthesis at $i$, queries can be supported on $BP(T)$ in $O(t(n))$ time with $o(n)$-bit auxiliary structures, when any $O(\log n)$-bit substring of the $BP(T)$ can be decoded in $t(n)$ time. 
Furthermore, it is known that the wide range of tree navigational queries on $T$ also can be answered in $O(t(n))$ time using $BP(T)$ with $o(n)$-bit auxiliary structures~\cite{navarro2014fully}: Here, each node is given and returned as the position of the open parenthesis that appended when the node is first visited during the construction of $BP(T)$ (for the full list of the queries, please refer to Table \RNum{1} in \cite{navarro2014fully}).
\\\\
\noindent\textbf{Rank and Select queries.}
Given a sequence $S = (s_1, s_2, \dots, s_n) \in \{0, \dots, \sigma-1\}^n$ of size $n$ over an alphabet of size $\sigma$, (a) $\rank{}_{S}(a, i)$ returns the number of occurrence
of $a \in \{0, \dots, \sigma-1\}$ in $(s_1, s_2, \dots, s_i)$, and (b) $\select{}_{S}(a, j)$ returns the first position of the $j$-th occurrence of $a \in \{0, \dots, \sigma-1\}$ in $S$ (in the rest of this paper, we omit $S$ if it is clear from the context). The following data structures are known, which can support both $\rank{}$ and $\select{}$ queries efficiently using succinct space~\cite{raman2007succinct,belazzougui2015optimal}:
(1) suppose $\sigma = 2$, and $S$ has $m$ $1$s. Then there exists a $(\log {n \choose m} + o(n))$-bit data structure that supports both $\rank{}$ and $\select{}$ queries in $O(1)$ time. The data structure can also decode any $O(\log n)$ consecutive bits of $S$ in $O(1)$ time, (2) there exists an $(n \log \sigma +o(n))$-bit data structure that can support both $\rank{}$ and $\select{}$ queries in $O(1)$ time, and (3) if $\sigma = O(1)$ and one can access any $O(\log n)$-length sequence of $S$ in $t(n)$ time, one can support both $\rank{}$ and $\select{}$ queries in $O(t(n))$ time using $o(n)$-bit auxiliary structures.
\\\\
\noindent\textbf{Range minimum and previous/next smaller value queries. }
Given a sequence $S = (s_1, s_2, \dots, s_n)$ of size $n$ from a total order with two positions $i$ and $j$ with $i \le j$, the \textit{range minimum query} $\rminq{}(i, j)$ on $S$ returns the position of the smallest element within the range $s_i, \dots, s_j$. Similarly, a \textit{range maximum query} $\rmaxq{}(i, j)$ on $S$ is defined to find the position of the largest element within the same range.

In addition, one can define \textit{previous (resp. next) smaller value queries} at the position $i$ on $S$, denoted as $\psv(i)$ (resp. $\nsv(i)$), which returns the nearest position from $i$ to the left (resp. right) whose value is smaller than $s_i$. If there is no such elements, the query returns $0$ (resp. $n+1$). 
One can also define \textit{previous (resp. next) larger value queries}, denoted as $\plv(i)$ (resp. $\nlv(i)$) analogously.

It is known that if $S$ is a permutation, $\rminq$, $\psv$, and $\nsv$ queries on $S$ can be answered in $O(1)$ time, given a BP of $\MinC(S)$ with $o(n)$ bit auxiliary structures~\cite{navarro2014fully, DBLP:journals/tcs/Fischer11}.
\\

\noindent\textbf{Tree Covering.}
Here, we briefly outline Farzan and Munro's~\cite{Farzan2014} tree covering representation and its application in constructing a succinct tree data structure. The core idea involves decomposing the input tree into {\it mini-trees} and further breaking them down into smaller units called {\it micro-trees}. These micro-trees can be efficiently stored in a compact precomputed table. The shared roots among mini-trees enable the representation of the entire tree by focusing only on connections and links between these subtrees. We summarize the main result of Farzan and Munro's algorithm in the following theorem.

\begin{theorem}[\cite{Farzan2014}]\label{TC}
For a rooted ordered tree with $n$ nodes and a positive integer $1 \leq \ell_1 \le n$,
one can decompose the trees into subtrees satisfying the following conditions:
(1) each subtree contains at most $2\ell_1$ nodes,
(2) the number of subtrees is $O(n/\ell_1)$,
(3) each subtree has at most one outgoing edge, apart from those from the root of the subtree.
\end{theorem}

See Figure~\ref{fig:separable} for an example. 
After decomposing the subtree as above, any node with an outgoing edge to a child outside the subtree is termed a \textit{boundary node}. The corresponding edge is referred to as the \textit{non-root boundary edge}. Each subtree has at most one boundary node and a non-root boundary edge. Additionally, the subtree may have outgoing edges from its root node, designated as \textit{root boundary edges}. For example, to achieve a tree covering representation for an arbitrary tree with $n$ nodes, Theorem~\ref{TC} is initially applied with $\ell_1 = \log^2 n$, yielding $O(n/\log^2 n)$ mini-trees. The resulting tree, formed by contracting each mini-tree into a vertex, is denoted as the \textit{tree over mini-trees}. This tree, with $O(n/\log^2 n)$ nodes, can be represented in $O(n/\log n) = o(n)$ bits through a pointer-based representation. Subsequently, Theorem~\ref{TC} is applied again to each mini-tree with $\ell_2 = \frac{1}{6} \log n$, resulting in a total of $O(n/\log n)$ micro-trees. The \textit{mini-tree over micro-trees}, formed by contracting each micro-tree into a node and adding dummy nodes for micro-trees sharing a common root, has $O(\log n)$ vertices and is represented with $O(\log\log n)$-bit pointers. Encoding the non-root/root boundary edge involves specifying the originating vertex and its rank among all children. The succinct tree representation, such as balanced parentheses (BP)~\cite{MunroR01}, is utilized to encode the position of the boundary edge within the micro-tree, requiring $O(\log \ell_2)$ bits. The overall space for all mini-trees over micro-trees is $O(n \log\log n/\log n) = o(n)$ bits. Finally, the micro-trees are stored with two-level pointers in a precomputed table containing representations of all possible micro-trees, demonstrating a total space of $2n+o(n)$ bits. By utilizing this representation, along with supplementary auxiliary structures that require only $o(n)$ bits of space, it is possible to perform fundamental tree navigation operations, such as accessing the parent, the $i$-th child, the lowest common ancestor, among many others, in $O(1)$ time~\cite{Farzan2014}.

\section{Succinct Representation of Baxter Permutation}\label{sec:rep}
In this section, we present a $(3n+o(n))$-bit representation for a Baxter permutation $\pi = (\pi(1), \dots, \pi(n))$ of size $n$. We begin by providing a brief overview of our representation. It is clear that the tree structure of $\MinC{}(\pi)$, along with the associated node labels can decode $\pi$ completely. 
However, the straightforward storage of node labels uses $\Theta(n \log n)$ bits, posing an efficiency challenge. 
To address this issue, 
we first show that when $\pi$ is a Baxter permutation, a two-stack based algorithm can be devised to traverse the nodes of $\MinC{}(\pi)$ according to the increasing order of their labels.
After that, we present a $(3n+o(n))$-bit representation that stores the information used throughout the algorithm, and show that the representation can decode $\MinC{}(\pi)$ with the labels of the nodes. 

\begin{algorithm}
\DontPrintSemicolon
\caption{Two-stack based algorithm}\label{alg:two-stack}
Initialize two empty stacks $L$ and $R$.\;
Visit $\phi(1)$ (i.e., the root of $\MinC(\pi)$).\;
\While{$i=2\dots n$}{
    \tcp*[h]{The last visited node is $\phi(i-1)$ by Lemma~\ref{lem:two-stack}.}\;
    \eIf{$\phi(i)$ is a left child of its parent}{
        \eIf{$\phi(i-1)$ has a left child}{
            Visit the left child of $\phi(i-1)$.\;
        }{
            Pop a node from stack $L$, and visit the left child of the node.\;
        }
        \If{$\phi(i-1)$ has a right child that has not yet been visited}{
            Push $\phi(i-1)$ to the stack $R$.\;
        }
    }(\tcp*[h]{$\phi(i)$ is a right child of its parent}){
        \eIf{$\phi(i-1)$ has a right child}{
            Visit the right child of $\phi(i-1)$.\;
        }{
            Pop a node from stack $R$, and visit the right child of the node.\;
        }
        \If{$\phi(i-1)$ has a left child that has not yet been visited}{
            Push $\phi(i-1)$ to the stack $L$.\;
        }
    }
}
\end{algorithm}

Note that our encoding employs a distinct approach compared to prior representations, as seen in references~\cite{dulucq1998baxter, dulucq1996stack, gawrychowski2015optimal, DBLP:conf/latin/JoK22}. These earlier representations store the tree structures of $\MinC{}(\pi)$ and $\MaxC{}(\pi)$ (or their variants) together, based on the observation that there always exists a bijection 
between $\pi$ and the pair of $\MinC{}(\pi)$ and $\MaxC{}(\pi)$  
if $\pi$ is a Baxter permutation~\cite{dulucq1996stack}.
We show that for any node in $\MinC{}(\pi)$, our representation allows to decode the labels of its parent, left child, and right child in $O(1)$ time using $o(n)$-bit auxiliary data structures. Using the previous representations that only store tree structures of $\MinC{}(\pi)$ and $\MaxC{}(\pi)$, these operations can take up to $\Theta(n)$ time in the worst-case scenario, even though tree navigation queries can be supported in constant time.

Now we introduce a two-stack based algorithm to traverse the nodes in $\MinC{}(\pi)$ according to the increasing order of their labels. Let $\phi(i)$ denote the node of $\MinC(\pi)$ with the label $i$. The algorithm assumes that we know whether $\phi(i)$ is left or right child of its parent for all $i \in \{2, \dots, n\}$.
\begin{figure}[t]
\centering
\includegraphics[scale=0.32]{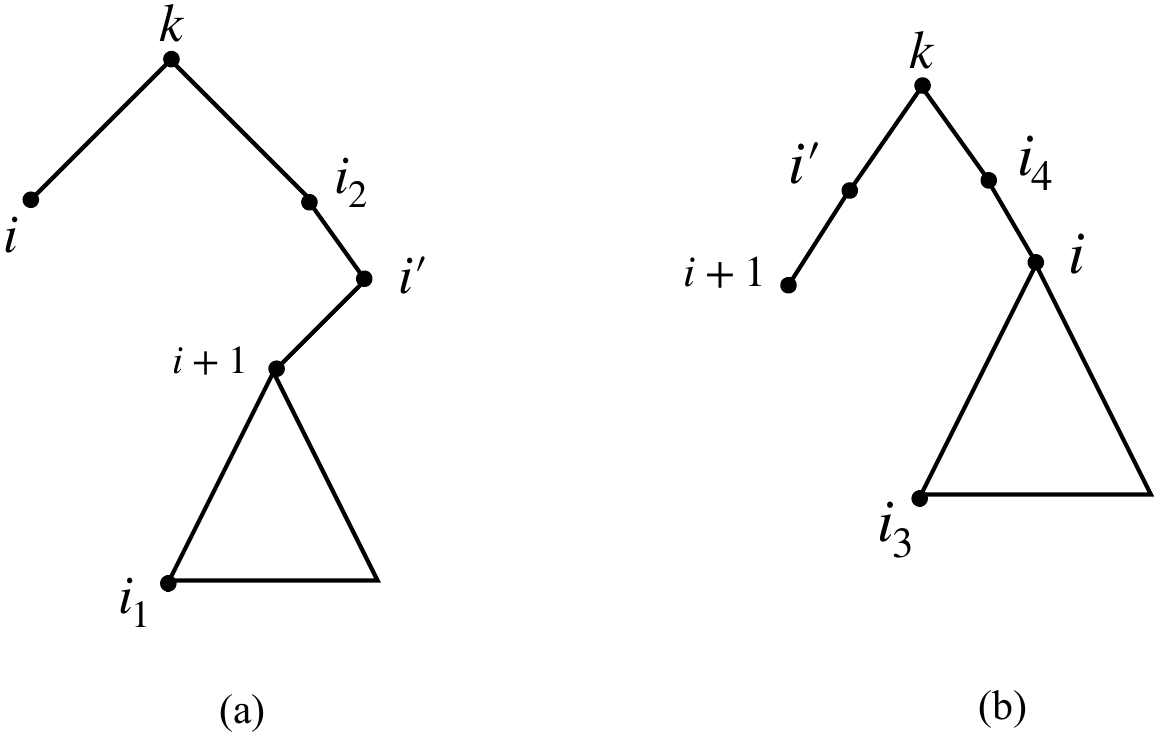}
\caption{(a) the case when $\phi(i)$ is in the left subtree of $\phi(k)$, and (b) the case when $\phi(i)$ is in the right subtree of $\phi(k)$.}
\label{fig:lemma1}
\end{figure}

The following lemma shows that if $\pi$ is a Baxter permutation, the two-stack based algorithm works correctly.
\begin{lemma}\label{lem:two-stack}
If $\pi$ is a Baxter permutation, the two-stack based algorithm on $\MinC{}(\pi)$ traverses the nodes according to the increasing order of their labels.
\end{lemma}
\begin{proof}
From Algorithm~\ref{alg:two-stack}, it is clear that we first visit the root node, which is $\phi(1)$. Then we claim that for any $i$, the two-stack based algorithm traverses the node $\phi(i+1)$ immediately after traversing $\phi(i)$, thereby proving the theorem.

Suppose not. Then we can consider the cases as (a) the left child of $\phi(i)$ exists, but $\phi(i+1)$ is not a left child of $\phi(i)$, or 
(b) the left child of $\phi(i)$ does not exist, but $\phi(i+1)$ is not a left child of the node at the top of $L$.
For the case (a) (the case (b) can be handled similarly), suppose
$\phi(i+1)$ is a left child of the node $\phi(i')$. Then 
thus $i' < i$ by the definition of $\MinC{}(\pi)$ and the condition of (a).
Now, let $\phi(k)$ be the lowest common ancestor of $\phi(i)$ and $\phi(i')$.
If $\phi(i)$ is in the left subtree of $\phi(k)$ (see Figure~\ref{fig:lemma1}(a) for an example), $k$ cannot be $i'$ from the definition of $\MinC(\pi)$. 
Then consider two nodes, $\phi(i_1)$ and $\phi(i_2)$, which are the leftmost node of the subtree rooted at node $\phi(i')$ and the node whose inorder is immediately before $\phi(i_1)$, respectively. 
Since $\phi(i_2)$ lies on the path from $\phi(k)$ to $\phi(i')$, we have $i+1 \le i_1$ and $k \le i_2 < i'$. 
Therefore, there exists a pattern $3\customdash{0.4em}14\customdash{0.4em}2$ induced by $i-i_2, i_1 - i'$, which contradicts the fact that $\pi$ is a Baxter permutation.

If $\phi(i)$ is in the right subtree of $\phi(k)$ (see~\ref{fig:lemma1}(b) for an example), $k$ cannot be $i$ from the definition of $\MinC(\pi)$. 
Consider two nodes, $\phi(i_3)$ and $\phi(i_4)$, which are the leftmost node of the subtree rooted at node $\phi(i)$ and the node whose inorder is immediately before $\phi(i_3)$, respectively. 
Since $\phi(i_4)$ lies on the path from $\phi(k)$ to $\phi(i)$, we have $i+1 < i_3$ ($i_3$ is greater than $i$ and cannot be $i+1$) 
and $k \le i_4 < i$. 
Therefore, there exists a pattern $3\customdash{0.4em}14\customdash{0.4em}2$ induced by $(i+1)-i_4, i_3 - i$, which contradicts the fact that $\pi$ is Baxter. 

The case when $\phi(i+1)$ is a right child of its parent can be proven using the same argument by showing that if the algorithm fails to navigate $\phi(i+1)$ correctly, the pattern $2\customdash{0.4em}41\customdash{0.4em}3$ exists in $\pi$.
\end{proof}

The representation of $\pi$ encodes the two-stack based algorithm as follows. 
First, to indicate whether each non-root node is whether a left or right child of its parent, we store a binary string $\lr[1, \dots n-1] \in \{l, r\}^{n-1}$ of size $n-1$ where $\lr[i] = l$ (resp. $\lr[i] = r$) 
if the node $\phi(i+1)$ is a left (resp. right) child of its parent. 
Next, to decode the information on the stack $L$ during the algorithm, we define an imaginary string of balanced parentheses $\lp[1 \dots n-1]$ as follows: After the algorithm traverses $\phi(i)$, $\lp[i]$ is 
(1) $\mbox{`('}$ if the algorithm pushes $\phi(i)$ to the stack $L$, (2) $\mbox{`)'}$
if the algorithm pops a node from the stack $L$, and (3) undefined otherwise.
We also define an imaginary string of balanced parentheses $\rp[1, \dots n]$ in the same way to decode the information on the stack $R$ during the algorithm. We use $`\{'$ and $`\}'$ to denote the parentheses in $\rp$.
Then from the correctness of the two-stack algorithm (Lemma~\ref{lem:two-stack}), and the definitions of $\lr$, $\lp$, and $\rp$, we can directly derive the following lemma:

\begin{lemma}\label{lem:lrseq}
For any $i \in \{1, \dots, n-1\}$, the following holds:
\begin{itemize}[itemsep=3pt]
    \item Suppose the node $\phi(i)$ is a leaf node. Then either $\lp[i]$ or $\rp[i]$ is defined. Also, $\lr[i]$ is $l$ (resp. $r$) if and only if $\lp[i]$ (resp. $\rp[i]$) is a closed parenthesis.
    \item Suppose the node $\phi(i)$ only has a left child. In this case, $\lr[i]$ is $l$ if and only if both $\lp[i]$ and $\rp[i]$ are undefined. Also, $\lr[i]$ is $r$ if and only if $\lp[i] = \mbox{`('} $ and $\rp[i] = \mbox{`\}'}$.
    \item Suppose the node $\phi(i)$ only has a right child. In this case, $\lr[i]$ is $l$ if and only if $\lp[i] = \mbox{`)'} $ and $\rp[i] = \mbox{`\{'}$. Also, $\lr[i]$ is $r$ if and only if both $\lp[i]$ and $\rp[i]$ are undefined.
    \item Suppose the node $\phi(i)$ has both left and right child. In this case, $\lr[i]$ is $l$ if and only if $\lp[i]$ is undefined and $\rp[i] = \mbox{`\{'}$. Also, $\lr[i]$ is $r$ if and only if $\lp[i] = \mbox{`('}$ and $\rp[i]$ is undefined.
\end{itemize}
\end{lemma}

\begin{figure}[bt]
\centering
\includegraphics[width=0.75\textwidth]{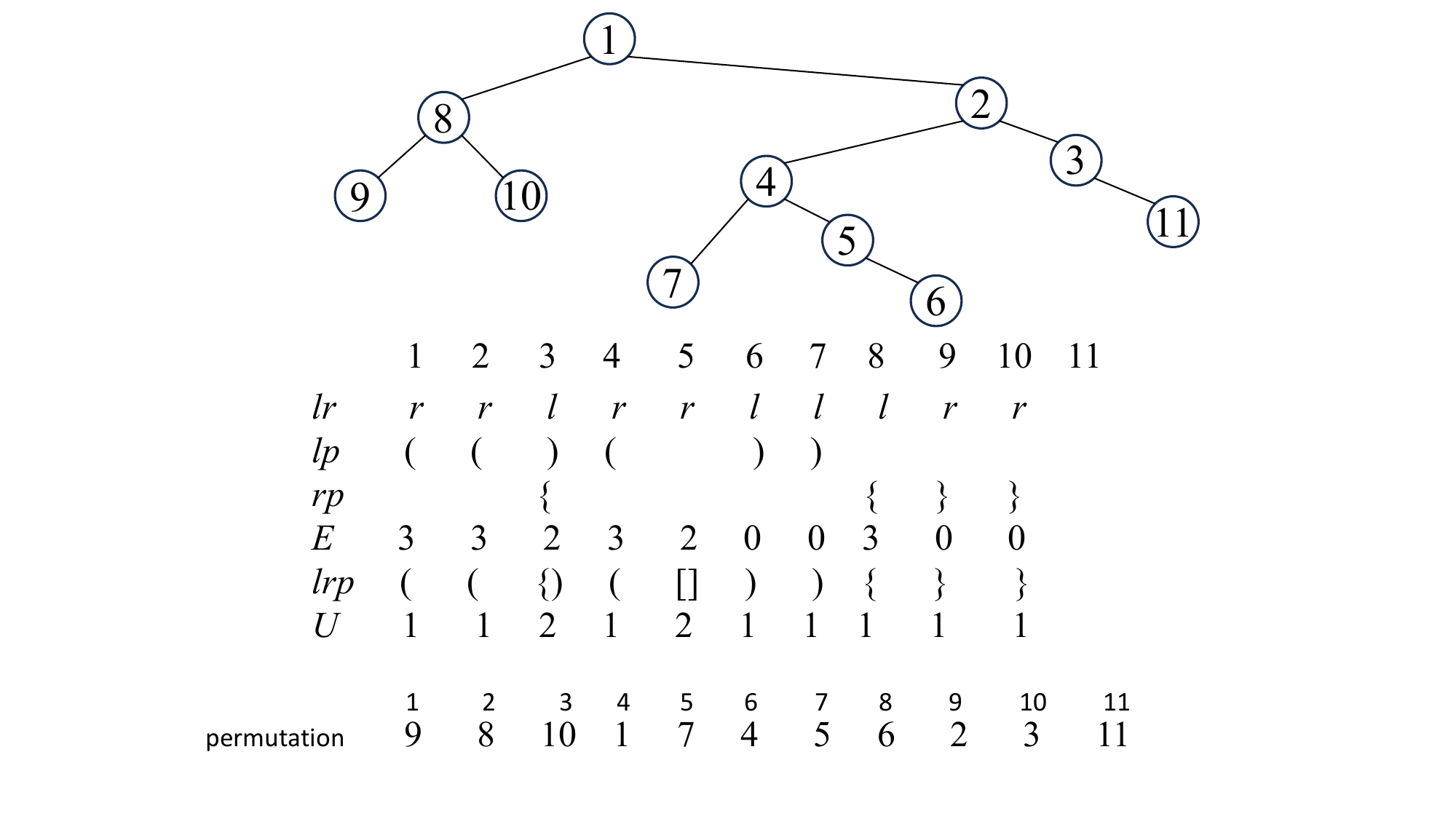}
\caption{An example of the representation of the Baxter permutation $\pi = (9, 8, 10, 1, 7, 4, 5, 6, 2, 3, 11)$. Note that the data structure maintains only $E$ and $lr$ along with $o(n)$-bit auxiliary structures. }
\label{fig:datastructure}
\end{figure}

To indicate whether each node $\phi(i)$ has a left and/or right child
we store a string $E \in \{0, 1, 2, 3\}^{n-1}$ of size $n-1$ where (a) $E[i] = 0$ if $\phi(i)$ is a leaf node (b) $E[i] = 1$ if $\phi(i)$ has only a left child, (c) $E[i] = 2$ if $\phi(i)$ has only a right child, and (d) $E[i] = 3$ if $\phi(i)$ has both left and right children. We store $E$ using $2n+o(n)$ bits, which allows support both $\rank$ and $\select$ operations in $O(1)$ time~\cite{belazzougui2015optimal}. Thus, the overall space required for our representation is at most $3n+o(n)$ bits ($2n$ bits for $E$, $n$ bits for $\lr$ along with $o(n)$-bit auxiliary structures).
From Lemma~\ref{lem:lrseq}, our representation can access $\lp[i]$ and $\rp[i]$ in $O(1)$ time by referring $\lr[i]$ and $E[i]$. 
Next, we show both $\findopen$ and $\findclose$ on $\lp$ and $\rp$ in $O(1)$ time using the representation.
We define an imaginary string $\lrp$ of length at most $2(n-1)$ over an alphabet of size $6$ 
that consists of three different types of parentheses $()$, $\{\}$, $[]$ constructed as follows. 
We first initialize $\lrp$ as an empty string and scan $\lr$ and $E$ from the leftmost position. 
Then based on Lemma~\ref{lem:lrseq}, whenever we scan $\lr[i]$ and $E[i]$, we append the parentheses to $\lrp$ as follows:
\begin{equation*}
\begin{cases}
  ( &\text{if $\lr[i] = r$ and $E[i] = 3$} \\
   ) &\text{if $\lr[i] = l$ and $E[i] = 0$}\\
   \{ &\text{if $\lr[i] = l$ and $E[i] = 3$}\\
  \} &\text{if $\lr[i] = r$ and $E[i] = 0$}\\
  (\} &\text{if $\lr[i] = r$ and $E[i] = 1$} \\
\{) &\text{if $\lr[i] = l$ and $E[i] = 2$}\\
  [] &\text{if (1) $\lr[i] = l$ and $E[i] = 1$, or (2) $\lr[i] = r$ and $E[i] = 2$}
\end{cases}
\end{equation*}

We store a precomputed table that has all possible pairs of $\lr$ and $E$ of size $(\log n)/4$ as indices. For each index of the table, it returns $\lrp$ constructed from the corresponding pair of $\lr$ and $E$. Thus, the size of the precomputed table is $O(2^{\frac{3}{4}\log n} \log n)= o(n)$ bits.

Additionally, we define an imaginary binary sequence $U \in \{1, 2\}^{n-1}$ of size $n-1$, where $U[i]$ denotes the number of symbols appended to $\lrp$ during its construction by scanning $\lr[i]$ and $E[i]$. Then by Lemma~\ref{lem:lrseq}, we can decode any $O(\log n)$-sized substring of $U$ starting from position $U[i]$ by storing another precomputed table of size $o(n)$ bits, indexed by all possible pairs of $\lr$ and $E$ of size $(\log n) /4$. Consequently, we can support both $\rank$ and $\select$ queries on $U$ by storing $o(n)$-bit auxiliary structures, without storing $U$ explicitly~\cite{belazzougui2015optimal}.

To decode any $O(\log n)$-sized substring of $\lrp$ starting from position $\lrp[i]$, we first decode a $O(\log n)$-sized substring of $E$ and $\lr$ from the position $i' = i - \rank_{U}(2, i)$ and decode the substring of $\lrp$ by accessing the precomputed table a constant number of times (bounded conditions can be easily verified using $\rank_{U}(2, i-1)$).
Thus, without maintaining $\lrp$, we can support $\rank$, $\select$, $\findopen$, and $\findclose$ 
queries on $\lrp$ in $O(1)$ time 
by storing $o(n)$-bit auxiliary structures~\cite{belazzougui2015optimal, chuang1998compact}. 
With the information provided by $\lrp$ and $U$, we can compute $\findopen(i)$ and $\findclose(i)$ operations on $\lp$ in $O(1)$ time as follows: To compute $\findopen(i)$, we compute $i_1 - \rank_{U}(2, i_1-1)$, where $i_1$ is the position of the matching `(' corresponding to $\lrp[i+\rank_{U}(2, i)]$.
For computing $\findclose(i)$, we similarly compute $i_2 - \rank_{U}(2, i_2)$, where $i_2$ corresponds to the position of the `)' corresponding to $\lrp[i+\rank_{U}(2, i-1)]$. Likewise, we can compute $\findopen(i)$ and $\findclose(i)$ operations on $\rp$ by locating the matching '\{' or '\}' in $\lrp$. In summary, our representation enables $\findopen$ and $\findclose$ operations on both $\lp$ and $\rp$ to be supported in $O(1)$ time without storing them explicitly.

Now we show that our representation is valid, i.e., we can decode $\pi$ from the representation. 
\begin{theorem}\label{thm:representation}
    The strings $\lr$ and $E$ give a $(3n+o(n))$-bit representation for the Baxter permutation $\pi = (\pi(1), \dots, \pi(n))$ of size $n$.
\end{theorem}
\begin{proof}
    It is enough to show that the representation can decode $\MinC{}(\pi)$ along with the associated labels. 
    For each non-root node $\phi(i)$, we can check $\phi(i)$ is either a left or right child of its parent by referring $\lr[i-1]$. Thus, it is enough to show that the representation can decode the label of the parent of $\phi(i)$. Without loss of generality, suppose $\phi(i)$ is a left child of its parent (the case that $\phi(i)$ is a right child of its parent is analogous). Utilizing the two-stack based algorithm and referring to Lemma~\ref{lem:lrseq}, we can proceed as follows: 
    If no element is removed from the $L$ stack after traversing $\phi(i-1)$ (this can be checked by referring $\lr[i]$ and $E[i]$), we can conclude that the parent node of $\phi(i)$ is indeed $\phi(i-1)$. Otherwise, the parent of $\phi(i)$ is the node labeled with $\findopen{}(i-1)$ on $\lp$ from the two-stack based algorithm.
\end{proof}
\begin{example}
Figure~\ref{fig:datastructure} shows the representation of the Baxter permutation $\pi = (9, 8, 10, 1, 7, 4, 5, 6, 2, 3, 11)$. Using the representation, we can access $\lp[3] = `)'$ by referring $\lr[3] = l$ and $E[3] = 2$ by Lemma~\ref{lem:lrseq}. Also, $\findopen(6)$ on $\lp$ computed by (1) computing the position of the matching `(' of the parenthesis of $\lrp$ at the position $i' = 6+ \rank_{U}(2, 6) = 8$, which is $5$, and (2) returning $5 - \rank_{U}(2, 5-1) = 4$. Note that $\lrp$ is not explicitly stored.
Finally, we can decode the label of the parent of $\phi(4)$ using $\findopen(3)$ on $\lp$ ($\phi(4)$ is the left child of its parent since $\lr[3] = l$), resulting in the value $2$. Thus, $\phi(2)$ is the parent of $\phi(4)$. 
\end{example}

\textbf{\noindent{Representation of alternating Baxter Permutation}}.
Assuming $\pi$ is an alternating permutation of size $n$, one can ensure that $\MinC(\pi)$ always forms a full binary tree by introducing, at most, two dummy elements $n+1$ and $n+2$, and adding them to the leftmost and rightmost positions of $\pi$, respectively~\cite{dulucq1996stack, cori1986shuffle}.
Specifically, we add the node $\phi(i+1)$ as the leftmost leaf of $\MinC(\pi)$ if $\pi(1) < \pi(2)$, Similarly, we add the node $\phi(i+2)$ as the rightmost leaf of $\MinC(\pi)$ if $\pi(n-1) > \pi(n)$.

Since no node in $\MinC(\pi)$ has exactly one child in this case, we can optimize the string $E$ in the representation of Theorem~\ref{thm:representation} into a binary sequence of size at most $n-1$, where $E[i]$ indicates whether the node $\phi(i)$ is a leaf node or not. Thus, we can store $\pi$ using at most $2n + o(n)$ bits. We summarize the result in the following corollary.

\begin{corollary}\label{cor:alter_represent}
    The strings $\lr$ and $E$ give a $(2n+o(n))$-bit representation for the alternating Baxter permutation $\pi = (\pi(1), \dots, \pi(n))$ of size $n$.
\end{corollary}

\section{Computing the BP sequence of Cartesian trees}\label{sec:query}

Let $\pi$ be a Baxter permutation of size $n$. In this section, we describe how to 
to compute $\pi(i)$ and $\pi^{-1}(j)$ for $i, j \in \{1,2,\ldots,n\}$ using the representation of Theorem~\ref{thm:representation}. 
First in Section~\ref{sec:inorder} we modify Cartesian trees so that inorders are assigned to all the nodes. 
Then we show in Section~\ref{sec:BP} we can obtain the BP sequence of $\MinC{}(\pi)$ from our representation.
By storing the auxiliary data structure of~\cite{navarro2014fully}, we can support tree navigational operations in Section~\ref{pre:BP2}. Finally, in Section~\ref{sec:maximum}, we show that our data structure can also support the tree navigational queries on $\MaxC{}(\pi)$ efficiently, which used in the results in the succinct representations of mosaic floorplans and plane bipolar orientations.

To begin discussing how to support $\pi(i)$ and $\pi^{-1}(j)$ queries, we will first show that the representation of Theorem~\ref{thm:representation} can efficiently perform a depth-first traversal on $\MinC{}(\pi)$ using its labels. We will establish this by proving the following lemma, which shows that three key operations, namely (1) $\lcl{}(i)$: returns the label of the left child of $\phi(i)$, (2) $\rcl{}(i)$: returns the label of the right child of $\phi(i)$, and (3) $\pl{}(i)$: returns the label of the parent of $\phi(i)$ on $\MinC{}(\pi)$, can be supported in $O(1)$ time.

\begin{lemma}\label{lem:dfsmin}
The representation of Theorem~\ref{thm:representation} can support $\lcl{}(i)$, $\rcl{}(i)$, and $\pl{}(i)$ in $O(1)$ on $\MinC{}(\pi)$ in $O(1)$ time.
\end{lemma}
\begin{proof}
The proof of Theorem~\ref{thm:representation} shows how to support $\pl{}(i)$ in $O(1)$ time.
Next, to compute $\lcl{}(i)$, it is enough to consider the following two cases according to Lemma~\ref{lem:lrseq}: (1) If $\lr[i] = l$ and $\lp[i]$ is undefined, $\lcl{}(i)$ is $i+1$, and (2) if $\lr[i] = r$ and $\lp[i] = \mbox{`('}$, we can compute $\lcl{}(i)$ in $O(1)$ time by returning $\findclose(i)$ on $\lp$.
Similarly, $\rcl{}(i)$ can be computed in $O(1)$ time using $\lr$ and $\rp$ analogously.
\end{proof}

Now we can compute $\phi(i+1)$ from $\phi(i)$ without using the two stacks in $O(1)$ time.
We denote this operation by $\nxt{}(i)$.
\begin{enumerate}[itemsep=3pt]
    \item If $\lr[i] = l$ and the left child of $\phi(i)$ exists, $\nxt{}(i)$ is the left child of $\phi(i)$.
    \item If $\lr[i] = r$ and the right child of $\phi(i)$ exists, $\nxt{}(i)$ is the right child of $v$.
    \item If $\lr[i] = l$ and the left child of $\phi(i)$ does not exist, 
    $\nxt{}(i)$ is the left child of $\phi(j)$ where $j = \findclose(i)$ on $\lp$.
    \item If $\lr[i] = r$ and the right child of $\phi(i)$ does not exist, 
    $\nxt{}(i)$ is the left child of $\phi(j)$ where $j = \findclose(i)$ on $\rp$.
\end{enumerate}

\subsection{Computing inorders}\label{sec:inorder}

First, we define the inorder of a node in a binary tree.
Inorders of nodes are defined recursively as follows.  We first traverse the left subtree of the root node and give inorders to the nodes in it,
then give the inorder to the root, and finally traverse the right subtree of the root node and give inorders.
In~\cite{navarro2014fully}, inorders are defined for only nodes with two or more children.
To apply their data structures to our problem, we modify a binary tree as follows.
For each leaf, we add two dummy children.
If a node has only right child,  
we add a dummy left child.
If a node has only left child, 
we add a dummy right child.  Then in the BP sequence $B$ of the modified tree, 
$i$-th occurrence of `)(' corresponds to the node with inorder $i$.
Therefore we can compute rank and select on `)(' in constant time using the data structure of~\cite{navarro2014fully}
if we store the BP sequence $B$ of the modified tree explicitly.
However, if we do so, we cannot achieve a succinct representation of a Baxter permutation.
We implicitly store $B$.  The details are explained next.

\subsection{Implicitly storing BP sequences}\label{sec:BP}

We first construct $B$ for $\MinC{}(\pi)$
and auxiliary data structures of~\cite{navarro2014fully} for tree navigational operations.
In their data structures, $B$ is partitioned into blocks of length $\ell$ for some parameter $\ell$,
and search trees called \textit{range min-max trees} are constructed on them.
In the original data structure, blocks are stored explicitly, whereas in our data structure,
they are not explicitly stored and temporarily computed from our representation.
If we change the original search algorithm so that an access to an explicitly stored block
is replaced with decoding the block from our representation, we can use the range min-max
trees as a black box, and any tree navigational operation
works using their data structure.  Because the original algorithms have constant query time,
they do a constant number of accesses to blocks.  If we can decode a block in $t$ time,
A tree navigational operation is done in $O(t)$ time.
Therefore what remains is, given a position of $B$, to extract a block of $\ell$ bits.

Given the inorder of a node, we can compute its label as follows.
For each block, we store the following.  For the first bit of the block, there are four cases:
(1) it belongs to a node in the Cartesian tree.
(2) it belongs to two dummy children for a leaf in the Cartesian tree.
(3) it belongs to the dummy left child of a node.
(4) it belongs to the dummy right child of a node.
We store two bits to distinguish these cases.
For case (1), we store the label and the inorder of the node using $\log n$ bits, and the information that
the parenthesis is either open or close using $1$ bit.
For case (2), we store the label and the inorder of the parent of the two dummy children,
and the offset in the pattern `(()())' of the first bit in the block.
For cases (3) and (4), we store the label and the inorder of the parent of the dummy node
and the offset in the pattern `()'.

To extract a block, we first obtain the label of the first non-dummy node in the block.
Then from that node, we do a depth-first traversal using $\lcl{}(i)$,
$\rcl{}(i)$, and $\pl{}(i)$, and compute a sub-sequence of $B$ for the block.
During the traversal, we also recover other dummy nodes.  Because the sub-sequence is of length $\ell$,
there are $O(\ell)$ nodes and it takes $O(\ell)$ time to recover the block.
To compute an inorder rank and select, we use a constant number of blocks.
Therefore it takes $O(\ell)$ time.
The space complexity for additional data structure is $O(n \log n/\ell)$ bits.
If we choose $\ell = \omega(\log n)$, the space is $o(n)$.

To support other tree operations including $\rminq$, $\nsv$, and $\psv$ queries on $\pi$, we use the original auxiliary data structures of~\cite{navarro2014fully}.
The space complexity is also $O(n \log n/\ell)$ bits. 

\subsection{Converting labels and inorders}\label{sec:labelinorder}

For the minimum Cartesian tree $\MinC{}(\pi)$ of Baxter permutation $\pi$, the label of the node with inorder $i$ is $\pi(i)$.
The inorder of the node with label $j$ is denoted by $\pi^{-1}(j)$.

We showed how to compute the label of the node with given inorder $i$ above.
This corresponds to computing $\pi(i)$.
Next we consider given label $j$, to compute the inorder $i = \pi^{-1}(j)$ of the node with label $j$.
Note that $\pi(i) = j$ and $\pi^{-1}(j) = i$ hold.

We use $\nxt{}(\cdot)$ to compute the inorder of the node with label $j$.
Assume $i\ell+1 \le j < (i+1)\ell$.
We start from the node $\phi(i\ell+1)$ with label $i\ell+1$ and iteratively compute $\nxt{}(\cdot)$ until we reach the node with label $j$.
Therefore for $i= 0, 1, \ldots, n/\ell$, we store the positions in the modified BP sequence for nodes $\phi(i\ell+1)$ using $O(n \log n/\ell)$ bits.
If $\nxt{}(i\ell+k)$ is a child of $\phi(i\ell+k)$, we can compute its position in the modified BP sequence
using the data structure of~\cite{navarro2014fully}.
If $\nxt{}(i\ell+k)$ is not a child of $\phi(i\ell+k)$, we first compute $p = \findclose(i\ell+k)$ on $\lp$ or $\rp$.
A problem is how to compute the node $\phi(p)$ and its inorder. 
To compute the inorder of $\phi(p)$, we use \emph{pioneers} of the BP sequence~\cite{GEARY2006231}.
A pioneer is an open or close parenthesis whose matching parenthesis belongs to a different block.
If there are multiple pioneers between two blocks, only the outermost one is a pioneer.
The number of pioneers is $O(n/\ell)$ where $\ell$ is the block size.
For each pioneer, we store its position in the BP sequence.
Therefore the additional space is $O(n \log n/\ell)$ bits.
Consider the case we obtained $p = \findclose(v)$.  If $v$ is a pioneer, the inorder of $\phi(p)$ is stored.
If $v$ is not a pioneer, we go to the pioneer that tightly encloses $v$ and $\phi(p)$,
obtain its position in the BP sequence, and climb the tree to $\phi(p)$.
Because $\phi(p)$ and the pioneer belong to the same block, this takes $O(\ell)$ time.
Computing a child also takes $O(\ell)$ time.
We repeat this $O(\ell)$ times until we reach $\phi(j)$.
Therefore the time complexity for converting the label of a node to its inorder takes $O(\ell^2)$ time.
The results are summarized as follows.
\begin{theorem}\label{thm:permutation_query}
    For a Baxter permutation $\pi$ of size $n$, $\pi(i)$ and $\pi^{-1}(j)$ can be computed
    in $O(\ell)$ time and $O(\ell^2)$ time, respectively, using a $3n+O(n \log n/\ell)$ bit data structure.
    This is a succinct representation of a Baxter permutation if $\ell = \omega(\log n)$. The data structure also can support the tree navigational queries in Section~\ref{pre:BP2} on $\MinC{}(\pi)$, $\rminq$, $\psv$, and $\nsv$ queries in $O(\ell)$ time.
\end{theorem}

Note that Theorem~\ref{thm:permutation_query} also implies that we can obtain the $(2n+o(n))$-bit succinct data structure of an alternating Baxter permutation of size $n$ that support $\pi(i)$ and $\pi^{-1}(j)$ can be computed in $O(\ell)$ time and $O(\ell^2)$ time, respectively, for any $\ell = \omega(\log n)$. 

\subsection{Navigation queries on Maximum Cartesian trees}\label{sec:maximum}
In this section, we show the representation of Theorem~\ref{thm:representation} can also support the tree navigational queries on $\MaxC{}(\pi)$ in the same time as queries on $\MinC{}(\pi)$, which will be used in the succinct representations of mosaic floorplans and plane bipolar orientations.

Note that we can traverse the nodes in $\MaxC{}(\pi)$ according to the decreasing order of their labels, using the same two-stack based algorithm as described in Section~\ref{sec:rep}.
Now, let $\phi'(i)$ represent the node in $\MaxC{}(\pi)$ labeled with $i$. We then define sequences $\lr$ and $E$ on $\MaxC{}(\pi)$ in a manner analogous to the previous definition (we denote them as $\lr'$, and $E'$, respectively). The only difference is that the value of $i$-th position of these sequences corresponds to the node $\phi'(n-i+1)$ instead of $\phi(i)$, since we are traversing from the node with the largest label while traversing $\MaxC{}(\pi)$. Then by Theorem~\ref{thm:representation} and \ref{thm:permutation_query}, it is enough to show how to decode any $O(\log n)$-size substring of $\lr'$ and $E'$ from $\lr$ and $E$, respectively. 

We begin by demonstrating that for any $i \in [1, \dots, n-1]$, the value of $\lr'[i]$ is $l$ if and only if $\lr[n-i]$ is $r$. As a result, our representation can decode any $O(\log n)$-sized substring of $\lr'$ in constant $O(1)$ time. Consider the case where $\lr[i]$ is $l$ (the case when $\lr[i] = r$ is handled similarly). In this case, according to the two-stack based algorithm, $\phi(i+1)$ is the left child of $\phi(i_1)$, where $i_1 \le i$. Now, we claim that $\phi'(i)$ is the right child of its parent. Suppose, for the sake of contradiction, that $\phi'(i)$ is a left child of $\phi'(i_2)$. Then $\phi(i+1)$ cannot be an ancestor of $\phi(i)$, as there are no labels between $i+1$ and $i$. Thus, $i_2 > i+1$, and there must exist a lowest common ancestor of $\phi'(i)$ and $\phi'(i+1)$ (denoted as $\phi'(k)$). At this point, $\phi'(i+1)$ and $\phi'(i_2)$ reside in the left and right subtrees rooted at $\phi'(k)$, respectively.  Now $i_3 \le i$ be a leftmost leaf of the subtree rooted at $\phi'(i)$. 
Then there exists a pattern $2\customdash{0.4em}41\customdash{0.4em}\customdash{0.4em}3$ induced by $(i+1)-k$ and $i_3-i_2$, which contradicts the fact that $\pi$ is a Baxter permutation.  

Next, we show that the following lemma implies that the representation can also decode any $O(\log n)$-size substring of $E'$ in $O(1)$ time from $E$ along with $\pi(1)$ and $\pi(n)$.

\begin{lemma}\label{lem:maxcartesian}
Given a permutation $\pi$, $\phi(i)$ has a left child if and only if $\pi^{-1}(i) > 1$ and 
$\pi(\pi^{-1}(i)-1) > i$. Similarly, $\phi(i)$ has a right child if and only if $\pi^{-1}(i) < n$ and 
$\pi(\pi^{-1}(i)+1) > i$.
\end{lemma}
\begin{proof}
We only prove that $\phi(i)$ has a left child if and only if $\pi^{-1}(i) > 1$ and 
$\pi(\pi^{-1}(i)-1) < i$ (the other statement can be proved using the same argument).
Let $i_1$ be $\pi(\pi^{-1}(i)-1)$. From the definition of the minimum Cartesian tree, if $\phi(i_1)$ is at the left subtree of $\phi(i)$, it is clear that $i_1 > i$. 
Now, suppose $i_1 > i$, but $\phi(i)$ does not have a left child. In this case, $\phi(i)$ cannot be an ancestor of $\phi(i_1)$. Thus, there must exist an element in $\pi$ positioned between $i_1$ and $i$, which contradicts the fact that they are consecutive elements.
\end{proof}

As a conclusion, the data structure of Theorem~\ref{thm:permutation_query} can support the tree navigational queries in Section~\ref{pre:BP2} on $\MaxC{}(\pi)$, and $\rmaxq$, $\psv$, and $\nsv$ queries in $\omega (\log n)$ time using $o(n)$-bit auxiliary structures from the results in Section~\ref{sec:BP}. We summarize the results in the following theorem.
\begin{theorem}\label{thm:max_cartesian}
    For a Baxter permutation $\pi$ of size $n$, The succinct data structure of Theorem~\ref{thm:permutation_query} on $\pi$ can support the tree navigational queries in Section~\ref{pre:BP2} on $\MaxC{}(\pi)$, $\rmaxq$, $\plv$, and $\nlv$ queries in $O(f_1(n))$ time for any $f_1(n) = \omega(\log n)$.
\end{theorem}

\section{Succinct Data Structure of Separable Permutation}\label{sec:separable_represent}
In this section, we present a succinct data structure for a separable permutation $\rho = (\rho(1), \dots, \rho(n))$ of size $n$
that supports all the queries in Theorem~\ref{thm:permutation_query} and \ref{thm:max_cartesian} in $O(1)$ time. 
The main idea of the data structure is as follows. 
It is known that for the separable permutation $\rho$, there exists a unique \textit{separable tree ($v-h$ tree)} $T_{\rho}$ of $n$ leaves~\cite{DBLP:journals/ipl/BoseBL98, DBLP:conf/dac/SzepieniecO80}, which will be defined later. Since $T_{\rho}$ is a labeled tree with at most $2n-1$ nodes, $O(n \log n)$ bits are necessary to store $T_{\rho}$ explicitly. Instead, we store it using a tree covering where each micro-tree of $T_{\rho}$ is stored as an index of the precomputed table that maintains all separable permutations whose separable trees have at most $\ell_2$ nodes, 
where $\ell_2$ is a parameter of the size of the micro-tree of $T_{\rho}$, which will be decided later.
After that, we show how to support the queries in Theorem~\ref{thm:permutation_query} and \ref{thm:max_cartesian} in $O(1)$ time using the representation, along with $o(n)$-bit auxiliary structures.

\subsection{Succinct Representation}
Given a separable permutation $\rho$ of size $n$, the separable tree $T_{\rho}$ of $\rho$ is an ordered tree with $n$ leaves defined as follows~\cite{DBLP:conf/dac/SzepieniecO80}:
\begin{itemize}[itemsep=1pt]
    \item Each non-leaf node of $T_{\rho}$ is labeled either $\oplus$ or $\ominus$. We call a $\oplus$ node as an internal node labeled with $\oplus$, and similarly, a $\ominus$ node as an internal node labeled with $\ominus$.
    \item The leaf node of $T_{\rho}$ whose leaf rank 
    (i.e., the number of leaves to the left) $i$ has a label $\rho(i)$. 
    In the rest of this section, we refer to it as the leaf $\rho(i)$.
    \item Any non-leaf child of $\oplus$ node is a $\ominus$ node. Similarly, any non-leaf child of $\ominus$ node is a $\oplus$ node.
    \item For any internal node $p \in T_{\rho}$, let $\rho_p$ be a sequence of the labels of $p$'s children from left to right, by replacing the label of non-leaf child of $p$ to the label of the leftmost leaf node in the rooted subtree at the node. Then if $p$ is a $\oplus$ (resp. $\ominus$ node), $\rho_p$ is an increasing (resp. decreasing) subsequence of $\rho$.
\end{itemize}

See Figure~\ref{fig:separable} for an example. Szepienic and Otten~\cite{DBLP:conf/dac/SzepieniecO80} showed that for any separable permutation of size $n$, there exists a unique separable tree of it with $n$ leaves. 

We maintain $\rho$ through the tree covering algorithm applied to $T_{\rho}$, with the parameters for the sizes of mini-trees and micro-trees as $\ell_1 = \log^2 n$ and $\ell_2 = \frac{\log n}{6}$, respectively. Here, the precomputed table maintains all possible separable permutations whose corresponding separable trees have at most $\ell_2$ nodes. Additionally, two special cases are considered: when the micro-tree is a singleton $\oplus$ or $\ominus$ node. Since any separable permutation stored in the precomputed table has a size at most $\ell_2$, there exist $o(n)$ indices in the precomputed table.

The micro-trees of $T_{\rho}$ are stored as their corresponding indices in the precomputed table, using $n \log (3+2\sqrt{2}) + o(n) \simeq 2.54n + o(n)$ bits in total. Furthermore, for each mini-tree (or micro-tree), we store $O(\log n)$-bit (or $O(\log \log n)$-bit) additional information to answer $\rho(i)$ and $\rho^{-1}(j)$ queries, which will be described in the next section.
In Section~\ref{sec:querysep}, we consider how to support $\rho(i)$ and $\rho^{-1}(j)$, as well as $\rminq$, $\rmaxq$, $\psv$, $\plv$, $\nsv$, and $\nlv$ queries on $\rho$ in $O(1)$ time, using the representation along with $o(n)$-bit additional auxiliary structures.

\begin{figure}[bt]
\centering
\includegraphics[width=0.6\textwidth]{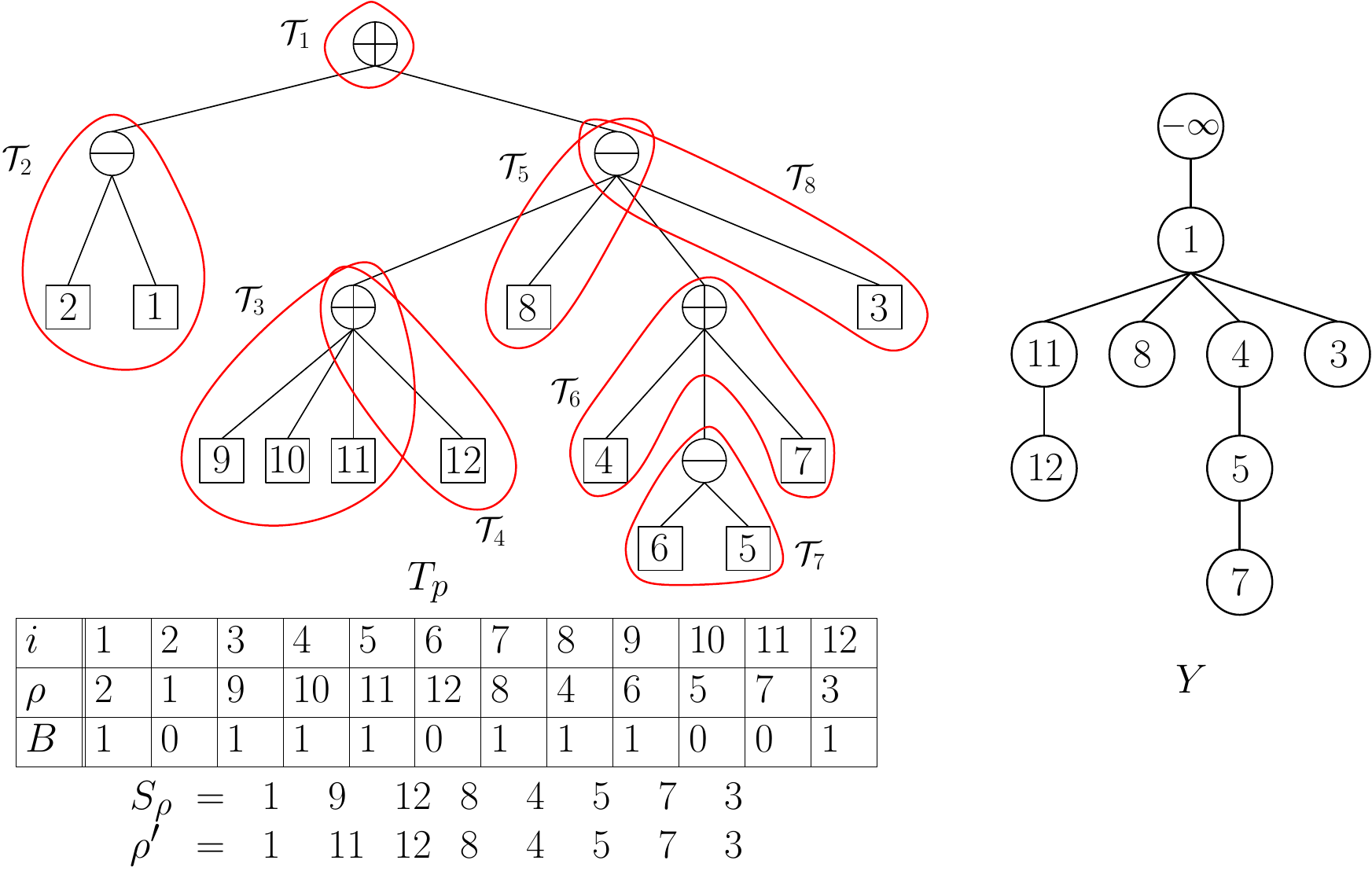}
\caption{An example of the representation of the separable permutation $\rho = (2, 1, 9, 10, 11, 12, 8, 4, 6, 5, 7, 3)$. 
Each tree within the red area represents a mini-tree of $T_{\rho}$ with $\ell_1 = 3$.}
\label{fig:separable}
\end{figure}

\section{Query algorithms on Separable Permutations}~\label{sec:querysep}
In this section, we consider the efficient support for $\rho(i)$, $\rho^{-1}(j)$, $\rminq$, and $\rmaxq$ on $\rho$ in $O(1)$ time, as well as $\psv$, $\plv$, $\nsv$, and $\nlv$ queries in $O(\log \log n)$ time, using the representation from the previous section along with $o(n)$-bit additional auxiliary structures.
We introduce the following proposition, derived directly from the definition of $T_{\rho}$ and the tree-covering algorithm. 

\begin{proposition}\label{prop:vhtree}
Given a separable tree $T_{\rho}$ of the separable permutation $\rho$, the following holds for each mini-tree (or micro-tree) $\tau$ of $T_{\rho}$ unless $\tau$ is either a singleton $\oplus$ or $\ominus$ node:
\begin{enumerate}[label=(\arabic*)]
    \item If $\tau$ has no boundary node, the labels of the leaves in $\tau$ cover all integers within a specified interval $[\min(\tau), \max(\tau)]$.
    
    \item If $\tau$ has a boundary node, the labels of the leaves in $\tau$ cover all integers in at most two disjoint intervals $[\min_l(\tau), \max_l(\tau)]$ and $[\min_r(\tau), \max_r(\tau)]$.
    
    Let $r_{\tau}$ denote the root of the mini-tree (or micro-tree) corresponding to the child of $\tau$ in the tree over mini-trees (or mini-tree over micro-trees) of $T_{\rho}$. Then, any leaf with a label in the interval $[\min_l(\tau), \max_l(\tau)]$ precedes $r_{\tau}$ in the preorder traversal of $T_{\rho}$. Similarly, any leaf with a label in the interval $[\min_r(\tau), \max_r(\tau)]$ follows $r_{\tau}$ in the preorder traversal of $T_{\rho}$.
    
    \item In case (2), when the boundary node of $\tau$ is a $\oplus$ node, $\min_r(\tau) = \max_l(\tau) + bsize(\tau)$, where $bsize(\tau)$ denotes the number of leaves in the rooted subtree at the boundary node of $\tau$. Conversely, if the boundary node of $\tau$ is a $\ominus$ node, $\min_l(\tau)  = \max_r(\tau) + bsize(\tau)$.
\end{enumerate}
\end{proposition}

\noindent\textbf{1. $\rho(i)$ and $\rho^{-1}(j)$ queries: } 
For each mini-tree $\mathcal{T}$ in $T_{\rho}$, we store $\min_l(\mathcal{T})$, $\max_l(\mathcal{T})$, $\min_r(\mathcal{T})$, and $\max_r(\mathcal{T})$ (if they exist). 
Additionally, we store the type ($\oplus$ or $\ominus$) of the root and the boundary node of $\mathcal{T}$.
This information can be stored in $O(\log n)$ bits per mini-tree, $O(n / \log n) = o(n)$ bits in total.
Also, for each micro-tree $\tau$ within the mini-tree $\mathcal{T}$, we maintain the same set of information relative to $\mathcal{T}$. Specifically, instead of storing the absolute value of $\min_l(\tau)$, we store the offset $\min_l(\tau) - \min_l(\mathcal{T})$ if the root of $\tau$ precedes the boundary node of $\mathcal{T}$ in the preorder traversal of $T_{\rho}$ (otherwise, we store the offset $\min_l(\tau) - \min_r(\mathcal{T})$). Similarly, we store $\min_r(\tau)$ as the offset $\min_r(\tau) - \min_l(\mathcal{T})$, if the boundary node of $\tau$ precedes to the boundary node of $\mathcal{T}$ in preorder traversal of $T_{\rho}$ (otherwise, we store the offset $\min_r(\tau) - \min_r(\mathcal{T})$). We also store $\max_l(\tau)$ and $\max_r(\tau)$ as offsets from either $\min_l(\mathcal{T})$ or $\min_r(\mathcal{T})$ in a similar manner. 
This information can be stored using $O(n \log\log n / \log n)  = o(n)$ bits in total. 
To answer the $\rho(i)$ query, we first find the mini-tree $\mathcal{T}$ and the micro-tree $\tau$ containing the leaf node $i$. This can be achieved in $O(1)$ time using $o(n)$-bit auxiliary structures for tree navigational queries~\cite{DBLP:journals/algorithmica/FarzanM14}. Let $i_{\tau}$ denote the leaf rank of the leaf node $i$ within $\tau$, which can be computed in $O(1)$ using the same structures. Additionally, let $\rho_{\tau}$ represent a separable permutation corresponding to $\tau$. We can then compute $\rho_{\tau}(i_{\tau})$ in $O(1)$ time using a precomputed table.
Therefore, by Proposition~\ref{prop:vhtree}, we can compute $\rho(i)$ in $O(1)$ time from either $\min_l(\mathcal{T})$ or $\min_r(\mathcal{T})$ along with the offsets stored at $\tau$ and $\rho_{\tau}(i_{\tau})$. Note that the offsets are used depending on the positions of the boundary nodes of $\mathcal{T}$ and $\tau$, as well as the leaf node $i$.

Next, we consider how to support the $\rho^{-1}(j)$ query in $O(1)$ time. 
Let $B[1, n]$ be a bit string of size $n$, where $B[p]$ is $1$ if and only if there exists a mini-tree $\mathcal{T}$ in $T_{\rho}$ such that $p$ corresponds to either $\min_l(\mathcal{T})$ or $\min_r(\mathcal{T})$. Since $B$ contains $O(n/\log^2 n) = o(n)$ $1$'s, we can represent $B$ using $o(n)$ bits, supporting $\rank$ and $\select$ queries in $O(1)$ time~\cite{raman2007succinct}. For each $1$ in $B$, we additionally store a pointer to the corresponding mini-tree (along with the information that indicates either the left or right part of the boundary node of the mini-tree) 
using $o(n)$ extra bits.
Similarly, for each mini-tree $\mathcal{T}$, we maintain two bit strings, $B^l_{\mathcal{T}}$ and $B^r_{\mathcal{T}}$, both of size $O(\log ^2 n)$, along with analogous pointers. More precisely, $B^l_{\mathcal{T}}$ (resp. $B^r_{\mathcal{T}}$) store the offsets of $\min_l(\tau)$ and $\min_r(\tau)$ of all micro-tree $\tau$ within $\mathcal{T}$ if they are stored as the offsets from $\min_l(\mathcal{T})$ (resp. $\min_r(\mathcal{T})$).
There exist $O(n/\log^2 n)$ such bit strings, and each of them contains $O(\log n) = o(\log^2 n)$ $1$'s since each mini-tree contains at most $O(\log n)$ micro-trees. 
Therefore, we can store these bit strings using $o(n)$ bits in total while supporting $\rank$ and $\select$ queries on them in $O(1)$ time.

To answer the $\rho^{-1}(j)$ query, we first identify the mini-tree that contains the leaf node $j$. From Proposition~\ref{prop:vhtree}, this can be done in $O(1)$ time by computing $\select (1, \rank(1, j))$. Suppose the leaf node $j$ is to the left of the boundary node of mini-tree $\mathcal{T}$. The case where the leaf node $j$ is to the right of the boundary node of $\mathcal{T}$ can be handled similarly. 
Next, we find the micro-tree within $\mathcal{T}$ that contains the leaf node $j$ in $O(1)$ time using $\rank$ and $\select$ queries on $B^l_{\mathcal{T}}$ with the offset $j - \min_l(\mathcal{T})$. Finally, based on the pointer corresponding to the $1$ in $B^l_{\mathcal{T}}$, we compute the leaf rank of the leaf node $j$ in $\tau$ by answering either $\rho_{\tau}^{-1}((j - \min_l(\mathcal{T})-\min_l(\tau)))$ or $\rho_{\tau}^{-1}((j - \min_l(\mathcal{T})-\min_r(\tau)))$. This can be done in $O(1)$ time using the precomputed table, allowing us to report the leaf rank of the leaf node $j$ in $T_{\rho}$ (i.e., $\rho^{-1}(j)$)  in $O(1)$ time.
\\\\
\noindent\textbf{2. $\rminq$ and $\rmaxq$ queries: }
We consider how to answer $\rminq$ queries on $\rho$ in $O(1)$ time, using the representation in Section~\ref{sec:separable_represent} of $T_{\rho}$ along with $o(n)$-bit auxiliary structures. Note that $\rmaxq$ queries can be similarly supported in $O(1)$ time with analogous structures.
To begin, for each separable permutation in the precomputed table, we store all possible answers to $\rminq$ queries on the permutation. Since each permutation in the table has $O(\log^2 n)$ distinct queries, the precomputed table can be maintained using $o(n)$ bits. Additionally, we employ a $O(n/\log^2 n)$-bit structure to answer $\rminq$ queries in $O(1)$ time on the subsequence $S_{\rho}$ of $\rho$ composed of the values $\min_l(\mathcal{T})$ and $\min_r(\mathcal{T})$ for all mini-trees $\mathcal{T}$ of $T_{\rho}$~\cite{DBLP:journals/tcs/Fischer11}. For each mini-tree of $T_{\rho}$ and its corresponding values in $S_{\rho}$, we store the bidirectional pointer between them, using $o(n)$ bits. 
Furthermore, for each mini-tree $\mathcal{T}$, let $S^l_{\mathcal{T}}$ (resp. $S^r_{\mathcal{T}}$) be the subsequences of $\rho$ composed of the values in 
$\min_l(\tau)$ and $\min_r(\tau)$ for all micro-trees within $\mathcal{T}$ whose corresponding leaves appear before (resp. after) the boundary node of $\mathcal{T}$ according to the preorder of $T_{\rho}$. 

To answer $\rminq{}(i, j)$ on $\rho$, we first determine the positions in $S_{\rho}$ corresponding to the mini-trees containing the leaves $i$ and $j$ in $O(1)$ time. Let $t_i$ and $t_j$ be these positions, 
respectively, and suppose $S_{\rho}(t_i) = \min_l(\mathcal{T}_i)$, and $S_{\rho}(t_j) = \min_l(\mathcal{T}_j)$ for some mini-trees $\mathcal{T}_i$ and $\mathcal{T}_j$ of $T_{\rho}$ (the other cases can be handled analogously).
Next, we find the positions in $S^l_{\mathcal{T}_i}$ and $S^l_{\mathcal{T}_j}$ corresponding to the micro-trees containing the leaves $i$ and $j$, denoted as $s_i$ and $s_j$, respectively. Let $s'_i$ and $s'_j$ be the positions in $\rho$ corresponding to the rightmost and leftmost leaves in the micro-trees that contain the leaves $i$ and $j$, respectively.
Then, $\rminq(i, j)$ is the position with the minimum value among the following five values:
(1) the $\rminq{}(i, s'_i)$-th value in $\rho$,
(2) the $\rminq{}(s'_j, j)$-th value in $\rho$,
(3) the $\rminq{}(s_i+1, |S^l_{\mathcal{T}_i}|)$-th value in $S^l_{\mathcal{T}_i}$,
(4) the $\rminq{}(1, s_j-1)$-th value in $S^l_{\mathcal{T}_j}$, and
(5) the $\rminq{}(t_i+1, t_j-1)$-th value in $S_{\rho}$.
Since values (1) and (2) can be computed in $O(1)$ time using the precomputed table with $\rho$ queries, and values (3), (4), and (5) can be computed in $O(1)$ time from the $\rminq$ queries on the sequences along with the pointers, $\rminq{}(i, j)$ can be computed in $O(1)$ time.
\\\\
\noindent\textbf{3. $\psv{}(i)$, $\plv{}(i)$, $\nsv{}(i)$, and $\nlv{}(i)$ queries: }
Here, we only consider how to answer $\psv{}(i)$ query in $O(1)$ time using the representation in Section~\ref{sec:separable_represent} of $T_{\rho}$ along with $o(n)$-bit auxiliary structures. Other queries also can be supported in $O(1)$ time with analogous structures.

First, we store all possible answers to $\psv{}$ queries on all separable permutations in the precomputed table using $o(n)$ bits.
For each mini-tree $\mathcal{T}$, we define two values $l_{\mathcal{T}}$ and $r_{\mathcal{T}}$ as follows: $l_{\mathcal{T}}$ is the label of the rightmost leaf node in $T_{\rho}$ that satisfies the following conditions (a) the leaf $l_{\mathcal{T}}$ is prior to the boundary node of $\mathcal{T}$ according to the preorder of $T_{\rho}$, and (b) $l_{\mathcal{T}}$ is smaller than the smallest label of the leaf node in the rooted subtree of $T_{\rho}$ at the boundary node of $\mathcal{T}$. 
Also, $r_{\mathcal{T}}$ is the label of the rightmost leaf node in $\mathcal{T}$.
Then we construct the subsequence $\rho'$ of $\rho$ composed to the values $l_{\mathcal{T}}$ and $r_{\mathcal{T}}$ for all mini-trees $\mathcal{T}$ within $T_{\rho}$. Now, we define a weighted ordered tree $Y$ as follows:
The root of $Y$ is a dummy node with a weight $-\infty$. 
Next, for each element $s$ in $\rho'$, we construct the node $y_s$ of $Y$ with a weight $s$, and store the pointer between the node $y_s$ and the leaf node in $T_{\rho}$ whose label is $s$. If there exists a value $s'$ in $\rho'$ where the $\psv$ at the position of $s$ in $\rho'$ is the position of $s'$, we define the parent of $y_s$ as $y_{s'}$. Otherwise, the parent of $y_s$ is the dummy node. The siblings in $Y$ are ordered based on the ordering of their corresponding values in $\rho'$. We then construct the data structure for answering \textit{weight-ancestor queries} on $Y$. This query is to find, given a query value and a node $y_p$, the nearest ancestor of $y_p$ in $Y$ whose weight is smaller than $p$. 
Since $Y$ has $O(n/\log^2n)$ nodes, where each node's weight is from a universe of size $n+1$, the weight-ancestor query on $Y$ can be answered in $O(\log \log n)$ time using $o(n)$ additional bits~\cite{DBLP:conf/esa/GawrychowskiLN14} by using van Emde Boas tree~\cite{DBLP:conf/focs/Boas75} for predecessor search.
Subsequently, for each mini-tree $\mathcal{T}$ of $T$, we maintain analogous structures with respect to the separable permutation corresponding to the leaves of $\mathcal{T}$ and the micro-trees within $\mathcal{T}$. In this case, we use $Y_{\mathcal{T}}$ to denote the weighted ordered tree on the permutation. The data structure for weight-ancestor queries on $Y_{\mathcal{T}}$ takes $o(\log^2 n)$ bits per each mini-tree $\mathcal{T}$ as each $Y_{\mathcal{T}}$ has $O(\log n)$ nodes and each node in $Y_{\mathcal{T}}$ has the weight from the universe of size $O(\log^2 n)$. Therefore, the overall space requirement is $o(n)$ bits in total~\cite{DBLP:conf/esa/GawrychowskiLN14, DBLP:conf/focs/Boas75}.

Let $\mathcal{T}$ and $\tau$ denote the mini-tree and the micro-tree containing the leaf $\rho(i)$, respectively. 
To answer $\psv(i)$, we identify the following three leaves as : (1) the rightmost leaf in $\tau$ to the left of the leaf $\rho(i)$ whose label is less than $\rho(i)$, (2) the leaf node $\rho(i')$ in $T_{\rho}$, where $\rho(i')$ is the rightmost leaf node to the left of the leaf $\rho(i)$, which is in $\mathcal{T}$ but not in $\tau$, and $\rho(i') < \rho(i)$, and (3) the leaf node $\rho(i'')$ in $T_{\rho}$, where $\rho(i'')$ is the rightmost leaf node to the left of the leaf $\rho(i)$, which is not in $\mathcal{T}$, and $\rho(i'') < \rho(i)$.
(1) can be computed in $O(1)$ time using the precomputed table, and (2) and (3) can be computed in $O(\log \log n)$ time using weight ancestor queries on $Y_{\mathcal{T}}$ and $Y$, respectively.
From the definition of $T_{\rho}$ and Proposition~\ref{prop:vhtree}, $\psv{}(i)$ corresponds to the leaf rank of one of the leaves among (1), (2), and (3). Therefore, we return the maximum leaf rank in $T_{\rho}$ among the leaves (1), (2), and (3) as the answer for $\psv{}(i)$ in $O(\log \log n)$ time.

\begin{example}
Figure~\ref{fig:separable} shows an example of the representation of the separable permutation $\rho = (2, 1, 9, 10, 11, 12, 8, 4, 6, 5, 7, 3)$. Here, due to space constraints, we only decompose $T_{\rho}$ into one-level mini-trees. This implies that we assume the operations within mini-trees can be answered using precomputed tables.
To answer $\rho(3)$, we first locate the mini-tree containing the leaf node with leaf rank $4$, which is $\mathcal{T}_4$. Since $\min_l(\mathcal{T}_4) = \min_l(\mathcal{T}_4) = 9$, we can answer $\rho(4) = 9 +1 = 10$. Next, to find $\rho^{-1}(6)$, we identify the mini-tree corresponding to $B[\select{}(1, (1, \rank{}(6)))]$, which is $\mathcal{T}_7$. Then, we return the leaf rank of the leaf node whose value is the $(6- \min_l(\mathcal{T}_7))$-th smallest value in $\mathcal{T}_7$, which is $9$ (here, the smallest value corresponds to the $0$-th smallest value).

Now, consider how to compute $\rminq{}(3, 9)$ on $\rho$. The leaves in $T_{\rho}$ with leaf ranks $3$ and $9$ are in $\mathcal{T}_3$ (corresponds to $S_{\rho}[2]$) and $\mathcal{T}_7$ (corresponds to $S_{\rho}[6]$), respectively. Thus, we compare three values: (1) $\rho(\rminq{}(3, 5)) = 9$, (2) $S_{\rho}[\rminq{}(3, 5)] = 4$, and (3) $\rho(\rminq{}(9, 9)) = 6$, and return the leaf rank of the leaf node $4$. Finally, consider how to compute $\psv{}(8)$ on $\rho$. In this case, there exists no leaf node at the left part of $\mathcal{T}_6$ whose value is smaller than $\rho(8) = 4$, and the weighted ancestor of $y_{l_{\mathcal{T}_6}} = y_4$ in $Y$ is $y_1$. Therefore, we return the leaf rank of the leaf node $1$ in $T_{\rho}$, which is $2$.
\end{example}

We summarize the results in the following theorem.
\begin{theorem}\label{thm:separable}
    For a separable permutation $\rho$ of size $n$, there exists a succinct data structure of Theorem~\ref{thm:permutation_query} on $\rho$ can support $\rho(i)$ and $\rho^{-1}(j)$ in $O(1)$ time. The data structure also supports $\rminq$, $\rmaxq$, $\psv$, $\plv$, $\nsv$, and $\nlv$ on $\rho$ in $O(\log \log n)$ time.
\end{theorem}

\section{Applications}\label{sec:application}
In this section, we introduce succinct representations for mosaic floorplans and plane bipolar orientations, both of which have a bijection to a Baxter permutation. In addition, we consider a succinct representation for slicing floorplans, a special type of mosaic floorplan that has a bijection to a separable permutation.
For these objects, we show that the queries considered in Theorem~\ref{thm:permutation_query} and \ref{thm:separable} on the corresponding permutations can be used to support specific navigation queries on them efficiently.

\begin{figure}[bt]
\centering
\includegraphics[width=.6\textwidth]{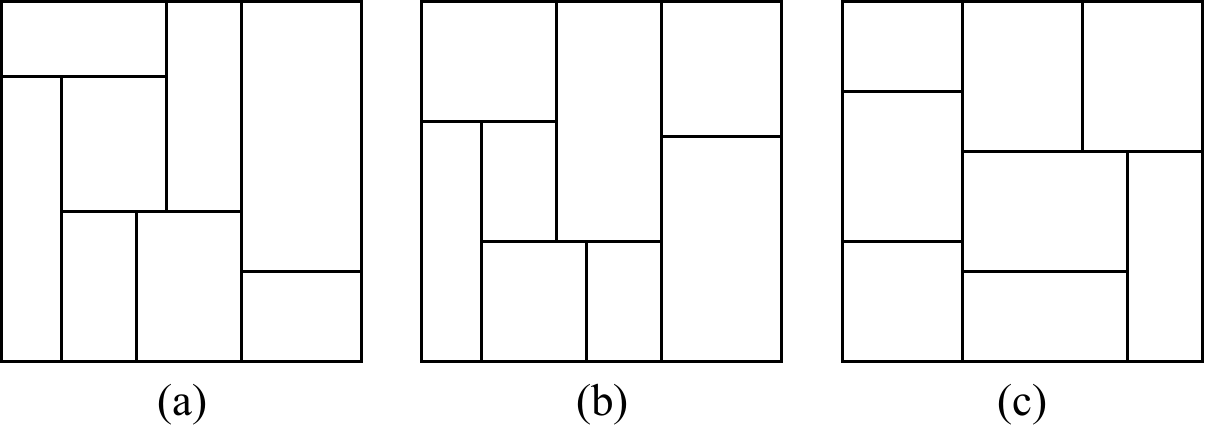}
\caption{(a) and (b) are  equivalent mosaic floorplans that are not slicing floorplans. (c) is a slicing floorplan.}
\label{fig:floorplan}
\end{figure}

\subsection{Mosaic and Slicing Floorplans}\label{sec:mosaic}
A \textit{floorplan} is a rectangular area that is divided by horizontal and vertical line segments. In a floorplan, these line segments can only create T-junctions, denoted as $\top$, $\bot$, $\vdash$, or $\dashv$ (see Figure~\ref{fig:floorplan} for an example). Each subdivided section within a floorplan is referred to as a \textit{block}, and the size of a floorplan is defined as the number of blocks it contains. 
For any two blocks labeled as $i$ and $j$, we define adjacency between them as follows:
(1) Block $i$ is directly above (or below) block $j$ if and only if block $i$ is above (or below) block $j$, and they share a horizontal line segment as a common boundary, and (2) block $i$ is directly left (or right) of block $j$ if and only if $i$ is left (or right) of $j$, and they share a vertical line segment as a common boundary.
\textit{Mosaic floorplans} are floorplans where their equivalence are defined as follows: Two mosaic floorplans, $F_1$ and $F_2$, of the same size are equivalent if and only if there exists a bijection $f$ between the blocks of $F_1$ and $F_2$ such that for any two blocks $i$ and $j$ in $F_1$, the adjacency relationship between $i$ and $j$ is the same as the adjacency relationship between $f(i)$ and $f(j)$ (see (a) and (b) of Figure~\ref{fig:floorplan} for an example). 

Also, a \textit{slicing floorplan} is a special type of mosaic floorplan. In slicing floorplans, the blocks in the floorplan are created by recursively dividing a single rectangle either horizontally or vertically into two smaller rectangles. As a result, slicing floorplans do not have 'pin-wheel' structure found in mosaic floorplans in general~\cite{DBLP:journals/dam/AckermanBP06}. See Figure~\ref{fig:floorplan} (c) for an example.
\newline
\newline
\noindent\textbf{Mosaic floorplans.} For mosaic floorplan $F$, Ackerman et al.~\cite{DBLP:journals/dam/AckermanBP06} defined two distinct orderings among the blocks in $F$: the \textit{top-left order} and the \textit{bottom-left order}. These orderings are derived from the following block-deletion algorithms.
In the top-left order, the first block is the top-left block in $F$. Next, for any $i > 1$, the $i$-th block is the top-left block of the deleting the $(i-1)$-th block from $F$ whose bottom-right corner is a $(\dashv)$ (resp. $(\bot)$)-junction, and shifting its bottom (resp. right) edge upwards (resp. leftwards) until the edge reaches the top (resp. left) boundary of $F$. Similarly, in the bottom-left order, the first block is the bottom-left block in $F$. Next, for any $i > 1$, the $i$-th block is the bottom-left block after deleting the $(i-1)$-th block from $F$ whose top-right corner is a $(\dashv)$ (resp. $(\top)$)-junction, 
and shifting its top (resp. right) edge downwards (resp. leftwards) until the edge reaches the bottom (resp. left) boundary of $F$. 

Ackerman et al.~\cite{DBLP:journals/dam/AckermanBP06} showed that any two equivalent mosaic floorplans have the same top-left and bottom-left orders. 
Furthermore, they showed that there exists a bijection between a mosaic floorplan $F$ of size $n$ and a Baxter permutation $\pi$ of size $n$. Specifically, for any $i \in [n]$, the $i$-th block according to the bottom-left order is the $(\pi(i))$-th block according to the top-left order (see Figure~\ref{fig:fp2bp} for an example).

In this section, we consider a succinct representation of mosaic floorplan $F$ of size $n$ that supports efficiently the following navigational queries. 
Here, each block $i$ is referred to as a $i$-th block of $F$ according to the bottom-left order. 
\begin{itemize}
    \item $\dabove{}(i,j)$: Returns true if and only if the block $i$ is directly above the block $j$.
    \item $\dbelow{}(i,j)$: Returns true if and only if the block $i$ is directly below the block $j$.
    \item $\dleft{}(i,j)$: Returns true if and only if the block $i$ is directly left-of the block $j$.
    \item $\dright{}(i,j)$: Returns true if and only if the block $i$ is directly right-of the block $j$.
    \item $\daboveset{}(i)$: Returns all $j$ where $\dabove{}(i,j)$ is true.
    \item $\dbelowset{}(i)$: Returns all $j$ where $\dbelow{}(i,j)$ is true.
    \item $\dleftset{}(i)$: Returns all $j$ where $\dleft{}(i,j)$ is true.
    \item $\drightset{}(i)$: Returns all $j$ where $\dright{}(i,j)$ is true.
\end{itemize}

\begin{figure}[bt]
\centering
\includegraphics[width=.25\textwidth]{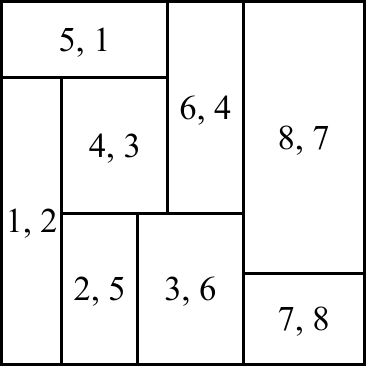}
\caption{A mosaic floorplan corresponds to $\pi=(2, 5, 6, 3, 1, 4, 8, 7)$. The first number of each block is a bottom-left deletion order, while the second number is a top-left deletion order.}
\label{fig:fp2bp}
\end{figure}

Our representation of $F$ simply represents its corresponding Baxter permutation $\pi$ using the data structure of Theorem~\ref{thm:permutation_query}. 
Now we show how to support the navigational queries on $F$ using the queries defined in $\pi$. 
We first introduce the following two lemmas essential for supporting the queries.
In this paper, we only provide the proof for Lemma~\ref{lem:nsvblock}
and note that the proof for Lemma~\ref{lem:nsvblock2} follows the same approach.

\begin{lemma}\label{lem:nsvblock}
Given a mosaic floorplan $F$ and its corresponding Baxter permutation $\pi$. Then for any two blocks $i$ and $j$ where $i$ is directly above $j$, the following holds:
(1) $\plv{}(i) = j'$ if and only if $j'$ is the rightmost block that directly below $i$, and (2) $\nsv{}(j) = i'$ if and only if $i'$ is the leftmost block that directly above $j$.
\end{lemma}
\begin{proof}
Here, we provide the proof of (1). Note that (2) can be proved using the same argument as the proof of (1). From the definitions of top-left and bottom-left orders,
it's clear that $\plv{}(i)$ is not a block directly above $j$ since the blocks are removed in a 
left-to-right order during any of two block-deletion algorithms. 
Additionally, $\pi(i) < \pi(j)$ since the bottom-right corner of $i$ (other than the rightmost one) cannot have a $(\dashv)$-junction. 
We now first prove that if $\plv{}(i) = j'$, then $j'$ is the rightmost block directly below $i$. Suppose not, then there exists a block $j'' \neq j'$ that is the rightmost block directly below $i$. This implies that $\pi(j'') > \pi(i)$ based on the top-left order. 
Then since $j'$ is positioned either to the left or below $j''$, it is clear that $j' < j'' < i$ is from the bottom-left order. This contradicts the fact that $\plv{}(i) = j'$.

Next, we prove if $j'$ is the rightmost block that directly below $i$, $\plv{}(i) = j'$.
We make the following claims: (1) the block $j'+1$ is directly above to the block $j'$, and
(2) For any two blocks $i_1$ and $i_2$ with $i_1 < i_2$ directly above to $j'$, $\pi(k) < \pi(i_2)$ for any $i_1 < k < i_2$. Note that the claims (1) and (2) together directly give the complete proof.
    
To prove (1), consider the block-deletion algorithm with the bottom-left order to determine the block $(j'+1)$. 
Since there are no blocks in $F$ to the right of the block $j'$ 
that share the top boundary with $j'$, the top-right corner of the block $j'$ 
forms a $(\dashv)$-junction. Thus, the next block is deleted after removing block $j'$ 
in the algorithm must share the boundary $h$ with $j'$.
Next, to prove (2), consider the block-deletion algorithm with the bottom-left order to determine the block $(i_1+1)$. 
Since $i_2$ is not equal to $i_1+1$, the top-right corner of $i_1$ must form a $(\dashv)$-junction. 
Thus, the block $(i_1+1)$ should be directly above the block $i_1$, implying $\pi(i_1) > \pi(i_1+1)$.
Furthermore, after deleting the block $i_1$, the block $(i_1+1)$ is directly to the left of the block $i_2$. 
Therefore, by iteratively applying the same argument up to block $(i_2-1)$, 
it follows that $\pi(k) < \pi(i_1)$ for any block $k$ in the set $\{(i_1+1), \dots, (i_2-1)\}$.
\end{proof}
\begin{lemma}\label{lem:nsvblock2}
Given a mosaic floorplan $F$ and its corresponding Baxter permutation $\pi$. Then for any two blocks $i$ and $j$ where $i$ is directly left of $j$, the following holds:
(1) $\nlv{}(i) = j'$ if and only if $j'$ is the bottommost block that directly right of $i$, and (2) $\psv{}(j) = i'$ if and only if $i'$ is the topmost block that directly left of $j$.
\end{lemma}

Now, we prove the main lemma, which shows that we can check the adjacency between two blocks using the data structure of Theorem~\ref{thm:permutation_query} on $\pi$.

\begin{lemma}\label{lem:directq}
    Given a mosaic floorplan $F$ and its corresponding Baxter permutation $\pi$, the following holds for any two blocks $i$ and $j$:
    \begin{enumerate}
        \item $\dabove{}(i, j)$ returns true if and only if 
        $\rmaxq{}(\nsv{}(j), i)=i$ and $\pi(i)<\pi(j)$.
        \item $\dbelow{}(i, j)$ returns true if and only if 
        $\rminq{}(i, \plv{}(j))=i$ and $\pi(i)<\pi(j)$.
        \item $\dleft{}(i, j)$ returns true if and only if 
        $\rmaxq{}(i, \psv{}(j))=i$ and $\pi(i)<\pi(j)$.
        \item $\dright{}(i, j)$ returns true if and only if 
        $\rminq{}(\nlv(j), i)=i$ and $\pi(i)>\pi(j)$.
    \end{enumerate}
\end{lemma}
\begin{proof}
We will focus on proving the statement for $\dabove{}$, as the other statements can be proved using a similar argument alongside Lemma~\ref{lem:nsvblock} and \ref{lem:nsvblock2}. 
When $\nsv(j) = i$, the statement is directly proved from Lemma~\ref{lem:nsvblock} 
since the block $\nsv{}(j)$ is the leftmost block directly above the block $j$ in this case.

Now, consider the case that $\nsv(j) < i$. We claim that for any block $i' \ge \nsv{}(j)$ positioned directly above $j$, the block $\nlv{}(i')$ is also directly above $j$ if and only if $\pi(j) > \pi(\nlv{}(i'))$. This claim proves the complete statement.
As indicated by Lemma~\ref{lem:nsvblock2}, $\nlv{}(i')$ represents the bottommost block located directly to the right of $i'$. Thus, we can consider only two possible cases:
(1) The block $\nlv{}(i')$ shares the same lower boundary as the block $i'$, indicating that block $\nlv{}(i')$ is directly above block $j$, or (2) The lower boundary of block $\nlv{}(i')$ lies below the lower boundary of block $i'$.
Consequently, our claim holds if we prove (2) occurs if and only if $\pi(\nlv{}(i')) > \pi(j)$. If the case (2) occurs, it implies that the bottom-right corner of $i'$ forms a $(\dashv)$-junction, and this, in turn, implies $\pi(\nlv{}(i')) > \pi(j)$ based on the top-left order. Conversely, if $\pi(\nlv{}(i')) > \pi(j)$, it is clear that $\pi(\nlv{}(i'))$ cannot be positioned above $\pi(j)$ according to the top-left order.
\end{proof}

By utilizing Lemma~\ref{lem:directq}, we can support $\dabove$, $\dbelow$, $\dleft$, and $\dright$ queries in $O(f_1(n))$ time where $f_1(n) = \omega(\log n)$, using $\pi(i)$ queries and  range minimum/maximum and previous/next larger/smaller queries on $\pi$.
Additionally, the proof presented in Lemma~\ref{lem:directq} implies that we can also handle $\daboveset{}(i)$, $\dbelowset{}(i)$, $\dleftset{}(i)$, and $\drightset{}(i)$ queries within $O(f_1(n))$ time per output. This is accomplished by iteratively computing one of the minimum/maximum or previous/next larger/smaller queries on $\pi$ and subsequently checking the value of $\pi$ at the computed position.

\begin{theorem}\label{thm:mosaic}
    Given any mosaic floorplan $F$ of size $n$, there exists a $(3n+o(n))$-bit representation of $F$ that supports
    \begin{itemize}
        \item $\dabove(i, j)$, $\dbelow(i, j)$, $\dleft(i, j)$ and $\dright(i, j)$ in $O(f_1(n))$ time, and
        \item $\daboveset(i)$, $\dbelowset(i)$, $\dleftset(i)$ and $\drightset(i)$ in $O(f_1(n)|output|)$ time,
    \end{itemize}
    for any $f_1(n) = \omega(\log n)$. 
    Here, $output$ counts the number of reported blocks of each query.
\end{theorem}

\noindent\textbf{Slicing floorplans.} Let $F$ be a slicing floorplan with $n$ blocks. Then there exists a unique separable permutation $\rho$ of size $n$ that corresponds to $F$~\cite{DBLP:journals/dam/AckermanBP06}. Furthermore, $\rho$ can be constructed using the same procedure for constructing a mosaic floorplan from its corresponding Baxter permutation. Consequently, all the preceding lemmas in this section also hold for $F$ and $\rho$. Thus, Theorem~\ref{thm:separable} and \ref{thm:mosaic} lead to the following corollary.

\begin{corollary}
    Given any slicing floorplan $F$ of size $n$, there exists a succinct representation of $F$ that supports $\dabove(i, j)$, $\dbelow(i, j)$, $\dleft(i, j)$ and $\dright(i, j)$ in $O(\log \log n)$ time, and
$\daboveset(i)$, $\dbelowset(i)$, $\dleftset(i)$ and $\drightset(i)$ in $O(\log \log n)$ time per output.
\end{corollary}

\subsection{Plane Bipolar Orientation}
A \textit{planar map} denoted as $M$ consists of vertices, edges, and faces, forming a connected graph embedded in the plane. This graph lacks edge intersections, with the outer face extending infinitely and all other faces bounded. When a map can be disconnected by removing a single vertex, it is termed \textit{separable}. A \textit{plane bipolar orientation}, denoted as $B$, refers to an acyclic orientation of a planar map. 
This orientation has a unique source vertex, denoted as $s$, which has no incoming edges, and a unique sink vertex, denoted as $t$, with no outgoing edges. Both $s$ and $t$ are situated on the outer face of the map. Given a planar map $M$ with its unique source $s$ and sink $t$, it is a known fact that $M$ possesses a bipolar orientation with $s$ as the source and $t$ as the sink if and only if the map formed by adding the edge $(s, t)$ to the outer face of $M$ is non-separable.

A study by Bonichon et al.~\cite{BonichonBF08} revealed the existence of a bijection between Baxter permutations $\pi$ of size $n$ and plane bipolar orientations $B$ with $n$ edges through an explicit construction algorithm. In this section, our goal is to support basic navigational queries on $B$ directly from $\pi$, i.e., without explicitly constructing $B$ from $\pi$. To this end, let us first briefly recall the algorithm to construct $B$ from $\pi$~\cite{BonichonBF08}. Given $\pi$, the algorithm intermediately constructs an embedded directed graph $G(\pi)$ (consisting of black and white vertices and straight edges between them) and finally modifies $G(\pi)$ to obtain $B$. Throughout this discussion, we assume that a Baxter permutation $\pi= (\pi(1), \dots, \pi(n))$ of size $n$ is represented by its plane diagram, i.e., the set of black vertices $b_i=(i, \pi(i))$ in $G(\pi)$. See Figure~\ref{fig:baxtertobipolar} for an example. White vertices (having half-integer coordinates) are added to $G(\pi)$ whenever there is an \textit{ascent} $t$ in $\pi$, i.e., $\pi(t) < \pi(t+1)$ for $1 \leq t \leq (n-1)$. Formally, for every ascent $t$, let $s_t= max\{\pi(i):i \leq t$ \text{and} $\pi(i) < \pi(t+1)\}$, then a white vertex $w_t$ is added at $(t+.5, s_t+.5)$. See Figure~\ref{fig:baxter_structure} for a generic structure of this construction. Observe that no vertices can exist in the light gray areas. We also add two special white vertices $w_0$ and $w_n$ at $(.5,.5)$ and $(n+.5,n+.5)$ respectively such that we can assume $\pi(0)=0$ and $\pi(n+1)=n+1$. Finally, between two vertices $x=(x_1,x_2)$ and $y=(y_1,y_2)$, a directed edge from $x$ to $y$ is added in $G(\pi)$ if and only if (1) $x_i \leq y_i$ for $1 \leq i \leq 2$, and (2) there does not exist any $z$ such that $x < z < y$. All these edges point to the North-East, and this completes the construction of $G(\pi)$. Now, Bonichon et al.~\cite{BonichonBF08} showed the following, 

\begin{theorem}[\cite{BonichonBF08}]\label{thm:baxtertobipolar}
For all Baxter permutations $\pi$, the embedded graph $G(\pi)$ is planar, bicolored (every edge joins a black vertex and a white one), and every black vertex has indegree and outdegree $1$.
\end{theorem}
Finally, the algorithm removes all black vertices to obtain a plane bipolar orientation $B$ with source $w_0$ and sink $w_n$. One of the most crucial features of this mapping algorithm is that every element of $\pi$ (or equivalently, every black vertex in $G(\pi)$) creates an edge in $B$, hence, the size of $\pi$ gets mapped to the number of edges of $B$, and vice-versa (Theorem 2.1 in~\cite{BonichonBF08}). Furthermore, the edge corresponding to $\pi(i)$ (henceforth, the edge $\pi(i)$) in $B$ is uniquely determined given any $i$. In light of the above discussion, let us now define the queries formally. Given a Baxter permutation $\pi = (\pi(1), \dots, \pi(n))$ of size $n$, and indices $i,j$ such that $1 \leq i<j \leq n$,
\begin{itemize}
    \item are the directed edges $\pi(i)$ and $\pi(j)$ adjacent in $B$?
    \item enumerate all the directed adjacent neighbors of the edge $\pi(i)$ in $B$?
\end{itemize}

\begin{figure}[bt]
\centering
\includegraphics[width=.4\textwidth]{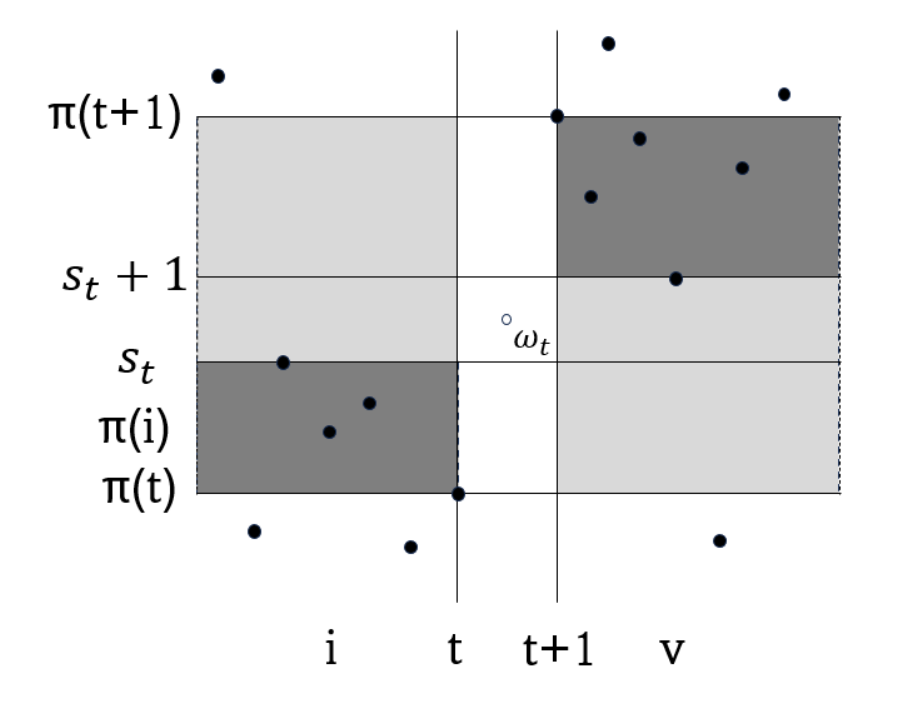}
\caption{Generic structure of the construction algorithm.}
\label{fig:baxter_structure}
\end{figure}

Here, given $i$ and $j$ such that $i<j$, we say that two edges in $B$ are \textit{adjacent} if the white endpoint of the edge $\pi(i)$ is the same as the start point of the edge $\pi(j)$. For example, in Figure~\ref{fig:baxtertobipolar}, $\pi(3)$ and $\pi(6)$ are adjacent whereas the edges $\pi(4)$ and $\pi(6)$ are not. In other words, these can also be thought of as a directed path of length $2$. Generalizing this, given $i$, the \textit{enumeration} query asks to list out all the adjacent edges of the edge $\pi(i)$. For example, the enumeration query on $\pi(3)$ (or equivalently, edge-labeled $2$) reports $4$ and $5$ respectively. As answers to these queries completely decode the entire graph, in graph data structure parlance, supporting them efficiently is of paramount importance. 

Before providing the query algorithms, let us prove some useful properties of any Baxter permutation $\pi$, its embedded graph $G(\pi)$, and their corresponding bipolar orientation $B$.

\begin{lemma}\label{lem:ascent}
    The edges $\pi(t)$ and $\pi(t+1)$ are adjacent in $B$ if and only if $t$ is an ascent.
\end{lemma}

\begin{proof}
For every ascent $t$ ($1 \leq t \leq (n-1)$) in $\pi$, i.e., $\pi(t) < \pi(t+1)$, the algorithm adds a white vertex $w_t$ at $(t+.5, s_t+.5)$ where $s_t= max\{\pi(i):i \leq t$ \text{and} $\pi(i) < \pi(t+1)\}$. Now, it is easy to see that there will be a directed edge from $(t,\pi(t))$ to $w_t=(t+.5, s_t+.5)$ as (a) $t \leq t+.5$, (b) $\pi(t) \leq s_t+.5$, and (c) there does not exist any vertex between $(t,\pi(t))$ and $(t+.5, s_t+.5)$. We can argue similarly about the existence of an edge from $w_t=(t+.5, s_t+.5)$ to $(t+1,\pi(t+1))$. When, finally, the black vertices $\pi(t)$ and $\pi(t+1)$ are removed from $G(\pi)$, they give rise to two edges that share the white vertex $w_t$ as the endpoint of the edge $\pi(t)$ and the starting point of the edge $\pi(t+1)$, making them adjacent. The other direction is easy to observe from the construction algorithm of $G(\pi)$.
\end{proof}

\begin{lemma}\label{lem:nlv_adjacent}
    The edges $\pi(i)$ and $\pi(\nlv(i))$ are adjacent. 
\end{lemma}

\begin{proof}
Suppose that $j=\nlv(i)$ for all $1 \leq i \leq (n-1)$, this implies that $\pi(k) < \pi(i)$ for all $ i < k < j$. This also implies that $\pi(j-1)$ and $\pi(j)$ form an ascent. As a consequence, by Lemma~\ref{lem:ascent}, there must exist a white vertex $w_{j-1}$ at $(j-.5, s_{j-1}+.5)$ where $s_{j-1}= max\{\pi(i):i \leq (j-1)$ \text{and} $\pi(i) < \pi(j)\}$. Furthermore, since $\pi(j-1)$ and $\pi(j)$ share $w_{j-1}$, they are also adjacent. All that remains to be shown is that there must exist an edge between $(i,\pi(i))$ and $w_{j-1}$. This is easy to verify as (a) $i \leq (j-.5)$ (as $j=\nlv(i)$) and (b) $\pi(i) \leq s_{j-1}+.5$ (by definition of $s_{j-1}$), and (c) there does not exist any vertex $q$ such that $(i,\pi(i)) < q < w_{j-1}$ as otherwise the condition $j=\nlv(i)$ is violated. 
\end{proof}

\begin{figure}[bt]
\centering
\includegraphics[width=.8\textwidth]{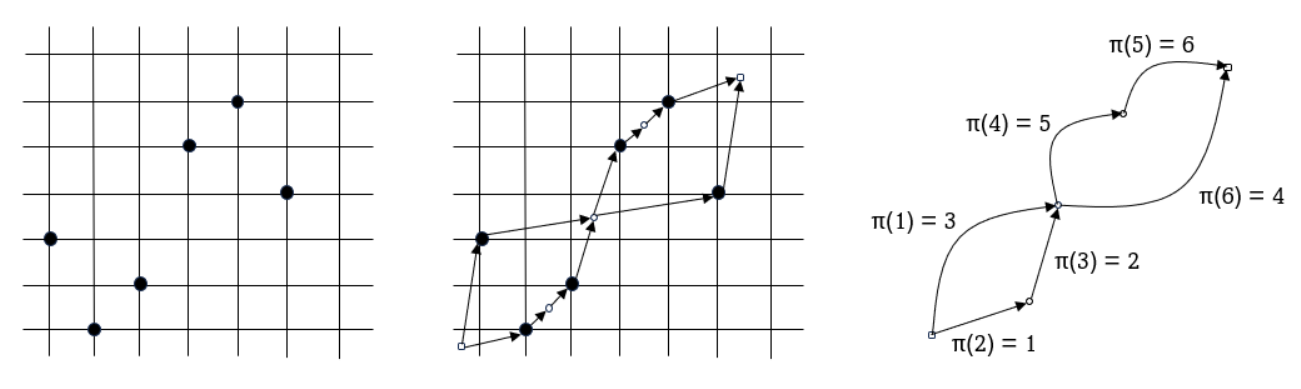}
\caption{The Baxter permutation $\pi= (3,1,2,5,6,4)$, the embedded directed graph $G(\pi)$, and the plane bipolar orientation $B$.}
\label{fig:baxtertobipolar}
\end{figure}

\begin{lemma}\label{lem:main_adjacent}
 The edges $\pi(i)$ and $\pi(k)$ are adjacent if and only if $\rminq{}(\nlv(i),k) = k$ for all $1 \leq i < \nlv(i) \leq k \leq n$ and $\pi(i) < \pi(k)$.   
\end{lemma}

\begin{proof}
    Suppose that $\pi(i)$ and $\pi(k)$ are adjacent. Then, by Lemma~\ref{lem:nlv_adjacent}, $\pi(i)$ and $\pi(\nlv(i))$ are adjacent. Suppose that $j=\nlv(i)$, then using the same argument as in Lemma~\ref{lem:nlv_adjacent} above, we know that there must exist an edge between $(i,\pi(i))$ and $w_{j-1}$. As $\pi(i)$ and $\pi(k)$ are adjacent, there must exist an edge from $w_{j-1}$ to $(k,\pi(k))$. This implies that (a) $ (j-.5) \leq k$, (b) $ s_{j-1}+.5 \leq \pi(k)$, and (c) there does not exist any vertex $q'$ such that $w_{j-1} < q' < (k,\pi(k))$. Therefore, $\pi(i) < \pi(k)$ and $\pi(k)$ admit the minimum value between (and including) $\pi(j)$ and $\pi(k)$, thus, $\rminq{}(\nlv(i),k) = k$.
    In the other direction, assume now that $\rminq{}(\nlv(i),k) = k$, $\pi(i) < \pi(k)$, and suppose that $j=\nlv(i)$, then from the construction of the embedded graph $G(\pi)$, it can be seen that both $(j,\pi(j))$ and $(k,\pi(k))$ share the same white vertex (as starting vertex), say $w$, that is created due to the ascent of $\pi(j-1) < \pi(j)$. Now, using similar arguments as in Lemma~\ref{lem:nlv_adjacent}, it is easy to see that there exists an edge between $(i,\pi(i))$ and $w$. Combining these two edges proves that $\pi(i)$ and $\pi(k)$ are adjacent.
\end{proof}

Note that both the observations stated in Lemma~\ref{lem:ascent} and Lemma~\ref{lem:nlv_adjacent} are special cases of Lemma~\ref{lem:main_adjacent}. Next, given $i$ and a Baxter permutation $\pi$ of size $n$, we consider the enumeration query which asks to list all the directed adjacent neighbors of the edge $\pi(i)$ in $B$. From the construction of the embedded graph $G(\pi)$ and the lemma~\ref{lem:nlv_adjacent}, it can be observed that the first neighbor (in left-to-right order in $\pi$) of $\pi(i)$ is the edge $\pi(\nlv(i))$. Let $t+1:=\nlv(i)$, then, we continue $\nsv$ queries (starting with $\nsv(t+1)$) and report as query answer $\pi(\nsv)$ values as long as $\pi(\nsv(r)) > \pi(i)$ and $\pi(\nsv(r+1)) < \pi(i)$ for some $t+1<r<n$. We claim that this procedure correctly reports all the directed neighbors of $\pi(i)$. See Figure~\ref{fig:baxter_structure} for a visual description of the proof. Note that for $\pi(i)$, the first ascent to its right is formed by $\pi(t)$ and $\pi(t+1)$ which in turn creates the white vertex $w_t$. From the construction algorithm, $(i, \pi(i))$ has an outgoing edge to $w_t$. As we start from $\pi(\nlv(i))=\pi(t+1)$, and continue to query smaller values successively, note that we are continuously moving down in the first quadrant (marked with dark grey color) of Figure~\ref{fig:baxter_structure}. All these vertices (call them $\pi(j)$s) must have an incoming edge from $w_t$ (hence, directed neighbors of $\pi(i)$) as long as $w_t \leq \pi(j)$ and there does not exist any $w_t < z < \pi(j)$. And, precisely, these conditions break down when we arrive at some $r$ such that $\pi(\nsv(r)) > \pi(i)$ and $\pi(\nsv(r+1)) < \pi(i)$ for some $t+1<r<n$. Furthermore, we can stop our algorithm here as we do not have to potentially look for any more neighbors of $\pi(i)$ as had such a neighbor existed, that would result in a $3\customdash{0.4em}14\customdash{0.4em}2$ pattern in $\pi$, which in turn violates the assumption that $\pi$ is Baxter. This concludes the description of our algorithms. Thus, the algorithm uses $\nlv$ and $\nsv$ queries for a combined $O(|neighbor|)$ time (here $|neighbor|$ counts the number of directed adjacent neighbors of $\pi(i)$, i.e., size of the output) along with using $\pi(i)$ queries each time to enumerate all the answers. Now, using Theorem~\ref{thm:permutation_query}, and Theorem~\ref{thm:max_cartesian}, 
we obtain the proof of the following theorem:
\begin{theorem}\label{thm:bipolar_main}
    Given a plane bipolar orientation $B$ with $n$ edges, there exists a $(3n+o(n))$-bit representation for $B$ that
    \begin{itemize}
        \item checks whether the directed edges $\pi(i)$ and $\pi(j)$ are adjacent in $B$ in $O(f_1(n))$ time, and
        \item enumerates the directed adjacent neighbors of the edge $\pi(i)$ in $B$ in $O(f_1(n))$~time per neighbor,
    \end{itemize}
    for any $f_1(n) = \omega(\log n)$.
\end{theorem}

\section{Future Work}\label{sec:future}
We conclude with the following concrete problems for possible further work in the future:
(1) Can we improve the query times of $\pi$ and $\pi^{-1}$ for Baxter permutations? (2) can we show any time/space trade-off lower bound for Baxter permutation similar to that of general permutation~\cite{golynski2009cell}? and (3) are there any succinct data structures for other pattern-avoiding permutations?

\bibliography{ref}
\end{document}